\documentclass[twoside,leqno]{article}

\usepackage[letterpaper]{geometry}

\usepackage{ltexpprt}
\usepackage[hidelinks]{hyperref}
\hypersetup{
    colorlinks=true,
    citecolor=blue,
}

\usepackage{sansmath}
\usepackage{amssymb}
\usepackage[reqno]{amsmath}
\usepackage{graphicx}
\usepackage{algorithm}
\usepackage{algorithmic}
\usepackage{xspace}
\usepackage{float}
\usepackage{enumitem}
\usepackage{tabularx}
\usepackage{xcolor}
\usepackage{subcaption}
\usepackage{graphicx} 
\usepackage{tikz}
\usepackage{wrapfig}

\newtheorem{observation}[lemma]{Observation}
\newtheorem{claim}[lemma]{Claim}
\newtheorem{definition}[lemma]{Definition}

\allowdisplaybreaks

\usepackage{thmtools}

\DeclareMathOperator{\E}{\mathbb E}
\DeclareMathOperator{\union}{\bigcup}

\newcommand{\Z}{\mathbb{Z}}

\newcommand{\calE}{{\mathcal E}}
\newcommand{\calF}{{\mathcal F}}

\newcommand{\calR}{{\mathcal R}}
\newcommand{\calS}{{\mathcal S}}

\newcommand{\eps}{\epsilon}

\newcommand{\sfd}{{\sansmath d}}

\usepackage{enumitem}
\setlist{itemsep=0pt,leftmargin=*}

\newcommand{\cand}{{\mathsf{cand}}}
\newcommand{\marked}{{\mathsf{mark}}}
\newcommand{\tsplit}{{\mathsf{split}}}	

\newcommand{\vol}{{\mathsf{vol}}}

\newcommand{\cI}{\mathcal{I}}

\newcommand{\srp}{\textsc{srp}}

\newcommand{\etal}{et al.\ }

\begin{document}

\title{\Large Improved Approximations for Unrelated Machine Scheduling}
\author{Sungjin Im\thanks{ Electrical Engineering and Computer Science, University of California, 5200 N. Lake Road, Merced CA 95344. \url{sim3@ucmerced.edu}. Supported in part by NSF grants CCF-1844939 and CCF-2121745.}  \and Shi Li\thanks{Department of Computer Science and Engineer, University at Buffalo, USA. \url{shil@buffalo.edu}. Supported in part by NSF grant CCF-1844890.}} 

\date{}

\maketitle


\begin{abstract}
    We revisit two well-studied scheduling problems in the unrelated machines setting where each job can have a different processing time on each machine. For minimizing total weighted completion time we give a 1.45-approximation, which improves upon the previous 1.488-approximation [Im and Shadloo SODA 2020]. The key technical ingredient in this improvement lies in a new rounding scheme that gives strong negative correlation with less restrictions. For minimizing $L_k$-norms of machine loads, inspired by [Kalaitzis \etal SODA 2017],  
    we give better approximation algorithms. In particular we give a $\sqrt {4/3}$-approximation for the $L_2$-norm which improves upon the former $\sqrt 2$-approximations due to [Azar-Epstein STOC 2005] and [Kumar \etal JACM 2009].
\end{abstract}

\section{Introduction}

Unrelated machines have been an elegant abstraction of various heterogeneous resource allocation problems. In the model there is a set $J$ of jobs to be scheduled on a set $M$ of machines. Job $j$ has an arbitrary processing time (or equivalently size) $p_{ij}$ on machine $i$. This general model has been extensively studied in various settings, such as offline, online, and stochastic settings. 
Further, due to the generality of the model, its study has led to discovering various algorithmic techniques. 

In this paper we revisit two well-studied offline scheduling/load balancing problems for unrelated machines, namely minimizing total weighted completion time and minimizing the $L_k$-norm of machine loads. We give improved approximation algorithms for both. 

\subsection{Previous Work} We describe the previous results related to the two problems we study.
\medskip

\noindent{\bf Total Weighted Completion Time.} \ 
This problem is denoted as $R||\sum_j w_j C_j$ using the classic common three-field notation.
In this problem, a machine can process at most one job at each time and all jobs must be completed non-preemptively. The problem is known to be strongly NP-hard and APX-hard \cite{hoogeveen2001non}. Schulz and Skutella gave a $1.5 + \eps$-approximation by using a time-indexed LP  and rounding the LP solution. The approximation ratio was improved to 1.5 by Skutella \cite{Skutella01} and Sethuraman and Squillante \cite{sethuraman1999optimal} independently, which both used convex programming. Breaking the barrier of 1.5 had since been widely open  \cite{chekuri2004approximation,schulz2002scheduling,kumar2008minimum,sviridenko2013approximating,schuurman1999polynomial}. 

The open problem was resolved by the breakthrough work by Bansal \etal \cite{BansalSS16}. All the previous works had used an independent rounding where each job $j$ is assigned to machine $i$ with probability $x_{ij}$ independently, given a fractional solution satisfying $\sum_{i \in M} x_{ij} = 1$. Even if all machines are identical, independent rounding only gives a $1.5$-approximation. To break the barrier, they introduced a novel rounding scheme
yielding certain strong negative correlations among some jobs grouped together. The weighted completion time objective is highly affected by how much jobs delay other jobs. To reduce the extra delay caused by randomized rounding, ideally we want to have a strong negative correlation for any pair of jobs, meaning two jobs $j$ and $j'$ are assigned to machine $i$ simultaneously with probability strictly less than $x_{ij} \cdot x_{ij'}$. Although this is not possible for every pair of jobs simultaneously, they showed it is possible for some jobs disjointly grouped on each machine. 

Unfortunately, the approximation ratio in  \cite{BansalSS16} was $1.5 - 10^{-7}$, which was only slightly better than the barrier of $1.5$. The approximation ratio was improved to $1.5 - 1/6000$ by Li \cite{Li20}, who used a time-indexed LP instead of SDP but used the same rounding scheme as Bansal \etal used.  
Later, it was further improved to $1.488$ by Im and Shadloo \cite{IS20}. Their main contribution was obtaining a significantly stronger negative correlation. Their rounding scheme was iterative fair contention resolution, which was inspired by the work of Feige and Vondrak \cite{feige2008allocations}. In contrast, Bansal \emph{et al.}'s rounding was a clever modification of pipage rounding. \medskip

\noindent {\bf $L_k$-norms of Machine Loads.}\ In this problem each machine $i$'s load is measured as the total processing time of jobs assigned to machine $i$. Note that machine $i$'s load is equivalent to the maximum completion time of jobs scheduled on machine $i$ non-preemptively with no idle times. The objective is minimizing the $L_k$-norm of machine loads. The $L_k$-norm seeks to reduce the variance of the loads as well as their average; for further motivation see~\cite{azar2004all}. 

The most well-known case is the $L_\infty$-norm, which is exactly the makespan minimization. For the $L_\infty$-norm, there exists a classic 2-approximation  and it is known to be hard to approximate within a factor of 1.5 unless P = NP \cite{shmoys1993approximation}. 
For the $L_k$-norm, Awerbuch \etal \cite{awerbuch1995load} gave a $\Theta(k)$-approximation based on independent rounding. Later, Azar and Epstein \cite{azar2005convex} gave a breakthrough by showing a $2$-approximation by using a convex programming and rounding it using the same rounding scheme used by Shmoys and Tardos \cite{shmoys1993approximation} for the Generalized Assignment
Problem. In particular they gave a $\sqrt 2$-approximation for the $L_2$-norm. Kumar \etal further improved the approximation ratio for small values of $k$ using a novel dependent rounding \cite{kumar2009unified}.

\subsection{Our Results}

\begin{theorem} (Sections~\ref{sec:SNC}-\ref{sec:wC-analysis}) \label{thm:wC}
    There is an efficient randomized 1.45-approximation algorithm for the problem of minimizing weighted completion time on unrelated machines.
\end{theorem}

As mentioned before, this improves upon the current best 1.488-approximation due to Im and Shadloo \cite{IS20}.

\begin{theorem} (Section~\ref{sec:lk}) \label{thm:Lk}
    There is a randomized algorithm that is 
    $\alpha^{1/k}$-approximate for minimizing the $L_k$-norm of machine loads, where 
    $\alpha := \sup_{m \in (0, 1], y \in [0, 1]} g(m, y)$, and 
        $$g(m, y) := \frac{ m \cdot (1+ (1 - m)y)^k + (1- m) \cdot ((1- m)y)^k}{m \cdot 1^k + (1 - m) \cdot y^k}$$
\end{theorem}

By using a numerical tool, we can compute an upper bound on the claimed approximation ratio $\alpha^{1/ k}$ for small values of $k$. For the important special case of $k = 2$,  we can analytically show that $\alpha = 4/3$; see Appendix~\ref{sec:k-2}.

\begin{table}[H]
\centering
\begin{tabular}{c | c | c | c | c | c | c | c | c }
$k$ & 1.5 & 2 & 2.5 & 3 & 3.5 & 4 & 4.5 & 5 \\
\hline
\hline
Kumar \etal \cite{kumar2009unified} & 1.585 & $\sqrt 2 \approx 1.414$ & 1.381 & 1.372 & 1.382 & 1.389 & 1.41 & 1.436 \\
\hline
This paper & 1.085 & $\sqrt {4/3} \approx 1.155 $ & 1.214 & 1.265 & 1.309 & 1.349 & 1.384 & 1.415   
\end{tabular}
\caption{Approximation ratio for minimizing the $L_k$-norm of machine loads for each value of $k$. }
\label{table:k}
\end{table}

\subsection{Our Techniques} In this section we give an overview of our techniques.
\medskip

\noindent{\bf Total Weighted Completion Time.}\ For the problem of minimizing weighted completion time, our main contribution is a new strong negative correlation (SNC) scheme. In the scheme introduced by Bansal et al.\ \cite{BansalSS16}, which was used later by Li \cite{Li20} and Im and Shadloo \cite{IS20}, there is a vector $y \in [0, 1]^{M \times J}$ such that $\sum_{i \in M} y_{ij} = 1$ for every $j \in J$. For every $i \in M$, we are given a partition of the edges incident to $i$ into groups such that every group has total $y$ value at most 1. A rounding algorithm with strong negative correlation will assign each job $j \in J$ to a machine $i$ randomly, respecting the marginal probabilities (i.e., $j$ is assigned to $i$ with probability $y_{ij}$), non-positive correlation between events that two jobs are assigned to a same machine, and strong negative correlation between events that two jobs are assigned to a same group. 

A crucial difference between our setting for the SNC scheme and that of Bansal et al.\ is that our partitions of edges are \emph{fractional}. As a result, it is more natural for us to define $y$ as a vector in $[0, 1]^{U \times J}$, where $U$ is the set of all groups. Additionally, we are given a mapping $g: U \to M$ indicating which machine each group belongs to. In contrast to the setting of Bansal et al., we may have $y_{uj} > 0, y_{u'j} > 0$ for a job $j \in J$ and two distinct groups $u, u' \in U$ with $g(u) = g(u')$.  Similarly, our rounding algorithm assigns jobs to groups, respecting the marginal probabilities, non-positive correlation between events that two jobs are assigned to two groups belonging to the same machine, and strong non-negative correlation between events that two jobs $j$ and $j'$ are assigned to the same group $u$.  The strong non-negative correlation holds when $\sum_{u' \in U: g(u') = g(u)}y_{u'j} \leq \frac12$ and $\sum_{u' \in U: g(u') = g(u)}y_{u'j'} \leq \frac12$. Such a restriction is needed: In the case where $g$ maps all the groups to the same machine, without the restriction, we can not satisfy all the three requirements.  As our SNC setting is more flexible, we have a simpler reduction from the scheduling problem to the SNC scheme, and the loss it incurs is smaller. This is explained below. 

As we do not require the groups belonging to a same machine to be disjoint, a simple construction of the SNC instance from the scheduling problem guarantees the condition (*): The total $y$ value of edges in a group is at most 1.  In contrast, the algorithm in \cite{IS20} requires a sampling step to ensure the disjointness of groups for the same machine. As a result, the condition (*) is satisfied only in expectation. To address the issue, they filter $1/2$ fraction of the edges in a group and apply Chernoff bound. Even with this procedure, (*) can only hold with a reasonable probability smaller than 1. As a result, their algorithm and analysis are more involved and their loss is bigger due to the filtering and the application of Chernoff bound. Furthermore, due to the simplicity of our algorithm, we can afford to analyze the approximation ratio more carefully. 

Our rounding algorithm for the SNC scheme borrows some ideas from \emph{online correlated selection} which played a crucial role in  the recent breakthrough result of a better-than-1/2 competitive algorithm for the online edge-weighted bipartite matching problem due to Fahrbach et al.\ \cite{FHT20}. Initially, every job chooses \emph{two} candidate groups belonging to different machines randomly. Then to respect the marginal probability, we need to assign the job to one of the two candidate groups uniformly at random. We correlate this process across different jobs to create some strong negative correlation between events that two jobs are assigned to the same group, without introducing positive correlation between events that two jobs are assigned to the same machine. To do so, we generate \emph{segments} of the form $u$-$j$-$u'$-$j'$-$u''$, where $j, j' \in J, u, u', u'' \in U$ and all edges are candidate edges. We then correlate the assignment of $j$ and $j'$: With probability $1/2$, we assign $j$ to $u$ and $j'$ to $u'$; with probability $1/2$, we assign $j$ to $u'$ and $j'$ to $u''$. So, we will not assign both $j$ and $j'$ to $u$, creating a negative correlation between the two events. The segments are generated so that for every $u', j, j'$, the path $j$-$u'$-$j'$ appears in the middle of a segment with some reasonable probability.  Our new SNC scheme as well as the rounding algorithm may be of independent interest. \medskip

\noindent{\bf $L_k$-norms of Machine Loads.}\  For minimizing $L_k$-norms of machine loads, we use the same algorithm that was used by \cite{Ola2017unrelated} for minimizing total weighted completion time when all jobs have an equal weight-to-size ratio, a.k.a., uniform smith ratios. Their algorithm is essentially the same as the celebrated algorithm Shmoys and Tardos used for minimizing makespan. The novelty of the work 
\cite{Ola2017unrelated} primarily lies in the analysis. 

At a high level, the rounding views the fractional LP solution as a fractional matching between jobs and ``buckets.'' A machine may have multiple buckets, each of unit capacity. The rounding samples an integral matching. The main difficulty in analyzing the algorithm comes from the fact that individual job assignments are highly correlated. The key observation in \cite{Ola2017unrelated} is the outcome satisfies a certain structural property. Then, the instance is sequentially modified into simpler instances keeping the property satisfied. We borrow this key idea. But, perhaps our analysis is (modestly) simpler and therefore works for
all $L_k$-norms of machine loads.

\subsection{Other Related Work} We first discuss other related work on the total  weighted completion time objective. The problem can be solved in polynomial time if jobs have uniform weights or all jobs are unit sized \cite{horn1973minimizing,bruno1974scheduling}. The problem admits a PTAS if the number of machine is fixed \cite{lenstra1990approximation}. As mentioned before, Kalaitzis \etal \cite{Ola2017unrelated} gave a 1.21-approximation when all jobs have the same ratio of $w_{ij}$ and $p_{ij}$ on each machine \cite{Ola2017unrelated}. When machines are   identical ($P||\sum_{j}w_jC_j$) or uniformly related ($Q||\sum_{j}w_jC_j$), the problem is NP-hard \cite{garey2002computers} but admits PTASes \cite{afrati1999approximation,skutella2000ptas,chekuri2001ptas}. If there is only one machine, the smith rule of ordering jobs in decreasing ratio of their weight to size is known to be optimal.

All hither-to-discussion assumed jobs arrive at time 0. If jobs arrive at different times, 
the problem is NP-hard even when there is only one machine \cite{lenstra1977complexity}. 
 When machines are unrelated ($R|r_j|\sum_j w_j C_j$), 2-approximation \cite{schulz2002scheduling,Skutella01} had been the best until it was improved to  1.869-approximation by \cite{im2016better}. 
 PTASes exist when machines are identical or uniformly related 
  \cite{afrati1999approximation,chekuri2001ptas}. The flow time objective is  significantly more challenging to approximate than the completion time objective. For minimizing total flow time, i.e, $R|r_j|\sum_j (C_j - r_j)$, a poly-logarithmic approximation was given by  \cite{BansalK15} and it was recently improved by \cite{Bansal-flow22}. 
For the special case of restricted assignment case with release times, see \cite{GargK07,GargKM08}.

For the makespan objective, i.e., $R||\max_j C_j$, there have been attempts to give better than 2-approximations for certain special cases \cite{svensson2012santa,ebenlendr2008graph}. The dual objective of maximizing the minimum total load of all machines has also been studied extensively~\cite{bansal2006santa,asadpour2010approximation,asadpour2008santa,chakrabarty2009allocating,feige2008allocations}. 
Kumar \etal gave a unified approach to  optimize multiple objectives simultaneously \cite{kumar2008minimum}. Bonifaci and Wiese showed a PTAS when there are a fixed number of machine types for minimizing $L_k$-norms of machine loads \cite{few-types-lk}.
Chakrabarty and Swamy initiated studying  the more general symmetric norm \cite{chakrabarty2019approximation}, which strictly generalized the $L_k$-norm of machine loads. The state-of-the-art approximation ratio for the problem is $2+ \eps$ \cite{ibrahimpur2021minimum}.

\subsection{Organization of the Paper} The rest of the paper is organized as follows. 
In  Section~\ref{sec:SNC} we introduce our new strong negative correlation scheme and give the algorithm. In Section~\ref{sec:wC-algorithm} we present our algorithm for $R||\sum_j w_jC_j$. We give its analysis in Section~\ref{sec:wC-analysis}, finishing the proof of Theorem~\ref{thm:wC}. In Section~\ref{sec:lk}, we prove Theorem~\ref{thm:Lk} by giving the algorithm for the problem of minimizing $L_k$-norm of loads and analyzing it. All missing proofs can be found in the appendix.  Most of the analysis of the algorithm for the weighted completion time problem is analytical. There are a few places where we use numerical tools to compute the maximum (minimum) of a single-variable function over an interval. We plot the two functions in  Appendix~\ref{appendix:wC-plotting}.

\section{Rounding Algorithm with Strong Negative Correlation}
\label{sec:SNC}
	In this section, we introduce our strong negative correlation (SNC) scheme.  In our setting, there is a set $M$ of machines, a set $J$ of jobs, a set $U$ of groups, a function $g: U \to M$ mapping groups to machines, and a vector $y \in [0, 1]^{U \times J}$ such that $y(u, J) \leq 1$ for every $u \in U$, and $y(U, j) = 1$ for every $j \in J$, where $y(u, J') = \sum_{j \in J'}y_{uj}$ for every $u \in U$ and $J' \subseteq J$, and $y(U', j) = \sum_{u \in U'} y_{uj}$ for every $U'\subseteq U$ and $j \in J$. We say a group $u \in U$ \emph{belongs to} the machine $g(u)$. We make the following definition, where $g^{-1}(i)$ is defined as $\{u \in U: g(u) = i\}$.
	\begin{definition}
		We say a machine $i \in M$ \emph{dominates} a job $j \in J$ if $y(g^{-1}(i), j) > \frac12$. 
	\end{definition}

	See Figure~\ref{fig:SNC-overview}(a) for the input of the SNC scheme. Our main theorem for the strong negative correlation scheme is the following. 
	\begin{theorem}
		\label{thm:SN}
		Given $M, J, U, g$ and $y$ as described above, there is an efficient randomized algorithm that outputs an assignment $\sigma: J \to U$ of jobs to groups satisfying the following properties.
		\begin{enumerate}[label=(\ref{thm:SN}\alph*)]
			\item \label{property:SN-marginal}{\bf (Marginal Probabilities)} For every $u \in U$ and $j \in J$ we have $\Pr[\sigma(j) = u] = y_{uj}$. 
			\item \label{property:SN-independence} {\bf (Non-Positive Correlation for a Same Machine)} For any two distinct jobs $j, j'$ and two (possibly identical) groups $u, u' \in U$ with $g(u) = g(u')$, we have 
			\begin{align*}
				\Pr[\sigma(j) = u, \sigma(j') = u'] \leq y_{uj}y_{u'j'}.
			\end{align*}
			\item \label{property:SN-SN} {\bf (Strongly Negative Correlation for a Same Group)} For any two distinct jobs $j, j' \in J$ and group $u \in U$ such that $g(u)$ does not dominate any of $j$ and $j'$, we have
			\begin{align*}
				\Pr\left[\sigma(j) = \sigma(j') = u\right] \leq (1-\eta)y_{uj}y_{uj'}, \text{ where } \eta = 0.1561.
			\end{align*}
		\end{enumerate}
	\end{theorem}

As we mentioned, in the scheme of Bansal et al.~\cite{BansalSS16}, the vector $y$ is in $[0, 1]^{M \times J}$ and there is a partition of edges incident to each machine $i$ into groups. This naturally gives a vector over $[0, 1]^{U \times J}$, for the set $U$ of all groups, with a crucial condition which our scheme does not have: a job can not be adjacent to two different groups belonging to the same machine, in the support of the vector.  Our $\eta = 0.1561$ may seem weaker than the counterpart in Im and Shadloo \cite{IS20}. However the flexibility of the scheme allows us to save more in the reduction from the scheduling problem.

		\subsection{Sketch of Rounding Algorithm} 

			We first give a sketch of the algorithm. To deliver the ideas more efficiently,  we assume no machines $i$ dominate any job $j$,  all the $y_{uj}$ values are very close to $0$, and we make many simplifying assumptions along the way.  Our algorithm consists of 4 steps; see Figure~\ref{fig:SNC-overview} for an illustration of the sketch.
			\begin{enumerate}
				\item For every job $j$, we choose two \emph{candidate groups} belonging to two different machines, such that a group $u$ is chosen with probability exactly $2y_{uj}$.  This is possible as we assumed no machines dominate any job. The process is independent for different jobs $j$. Let $H^\cand$ be the bipartite graph between $U$ and $J$ containing the edges $\{uj: u\text{ is a candidate group of } j\}$.  So, to respect the marginal probabilities, we need to assign each $j \in J$ to one of its two candidate groups uniformly at random. We do this in a coordinated way in succeeding steps.
				\item For every group $u \in U$, we consider its incident edges in $H^\cand$, or equivalently, the set of jobs choosing $u$ as a candidate group.  We pair these edges, and for simplicity we assume their cardinality is even so all edges can be paired. The pairing is chosen uniformly at random. Replacing each group $u$ in $H^\cand$ with multiple copies, each incident to one pair of edges, we obtain a graph $H^\tsplit$ where all vertices have degree $2$. Thus $H^\tsplit$ is the union of some cycles. 
				\item For simplicity,  we assume all cycles in $H^\tsplit$ have lengths being integer multiplies of $4$. We break each cycle in $H^\tsplit$ into paths of length $4$ with end points being groups, which we call \emph{segments}. There are two ways to do so and we choose one uniformly at random. So, each segment is of the form $u$-$j$-$u'$-$j'$-$u''$ with $u, u', u'' \in U$ and $j, j' \in J$. (We identify the copies of groups with their original groups in $U$). 
				\item For each segment $u$-$j$-$u'$-$j'$-$u''$, we perform the rounding operation: with probability $1/2$ we let $\sigma(j) = u$ and $\sigma(j') = u'$, and with the remaining probability $1/2$, we let $\sigma(j) = u'$ and $\sigma(j') = u''$. 
			\end{enumerate}
			
			\begin{figure}
			    \centering
			    \includegraphics[width=0.9\textwidth]{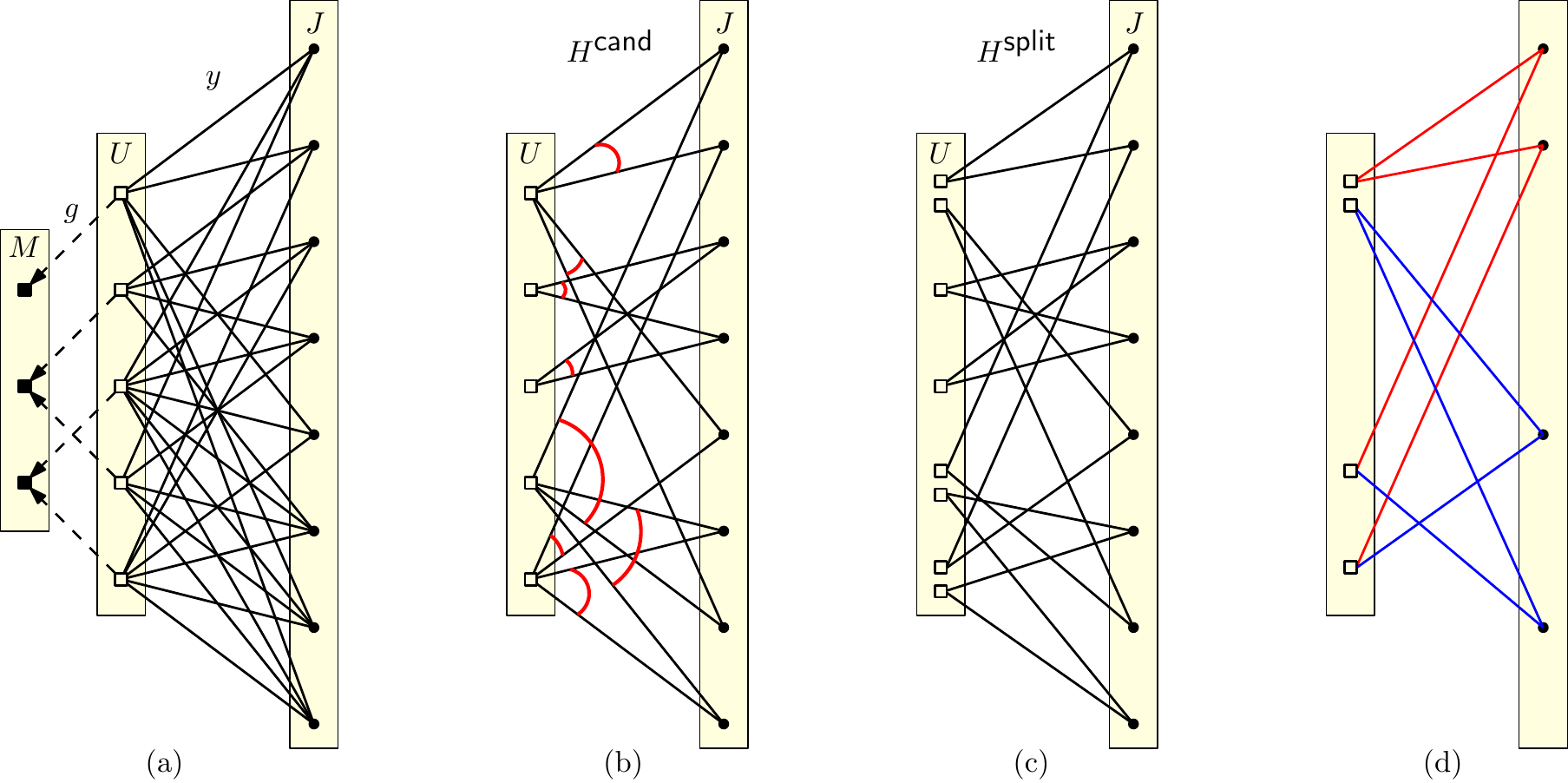}
			    \caption{The sketch of the algorithm for our SNC scheme. (a) depicts the input of the scheme. (b) shows the graph $H^\cand$, where every job $j \in J$ is incident to two groups belonging to two different machines. The pairing is indicated by the red arcs in (b). (c) shows the graph $H^\tsplit$. (d) gives one cycle-component in $H^\tsplit$ and a partition of the cycle into 2 segments, denoted by red and blue lines respectively.}
			    \label{fig:SNC-overview}
			\end{figure}
			
			This finishes the sketch of the algorithm.  Property~\ref{property:SN-marginal} is easy to see.  To see Property~\ref{property:SN-independence}, notice that Step 4 does not create a positive correlation over a same machine: We have $g(u) \neq g(u')$ and $g(u') \neq g(u'')$, due to that the two candidate groups of a job do not belong to a same machine.
			
			To analyze Property~\ref{property:SN-SN}, we show that two edges $uj$ and $uj'$ will be paired with a decent probability, conditioned on that they are both in $H^\cand$: When all the $y$ values tend to $0$, the degree of $u$ in $H^\cand$ follows a Poisson distribution with mean $2$.  Once their are paired, with probability $1/2$, $u$ will be the middle vertex of the segment containing $uj$ and $uj'$. If this happens, we do not assign both $j'$ and $j'$ to $u$ in Step 4, creating a negative correlation between $\sigma(j) = u$ and $\sigma(j') = u$.
			
			Many complications arise in the actual algorithm. Firstly, if a job $j$ is dominated by a machine $i$, we can not guarantee the two candidate groups for $j$ belong to different machines. In such a case, we do not impose the requirement.  Secondly, as our $y$ values may not tend to $0$, we need an extra filtering operation in Step 1: For every edge $uj \in H^\cand$ with $g(u)$ not dominating $j$, we \emph{mark} $uj$ with probability $\frac{1 - e^{-2y_{uj}}}{2y_{uj}}$, which approaches 1 as $y_{uj}$ tends to $0$.  This ensures that the probability an edge is marked is $1 - e^{-2y_{uj}}$, leading to the desired Poisson distribution in the analysis.  For every group $u$, we only pair the marked edges incident to it.   Thirdly, the cycles in $H^\tsplit$ may have lengths not being multiplies of $4$ (but the lengths are always even). We break such a cycle into segments of length $4$, and one segment of length $6$ satisfying certain properties.  Finally, $H^\tsplit$ may contain paths as edges may be unpaired; however they are easier to handle than cycles. 
			
		
			\subsection{Formal Description of Rounding Algorithm} \label{subsec:SNC-algorithm}
			
			Now we formally describe the rounding algorithm for Theorem~\ref{thm:SN}. Throughout the section, given a graph $H = (V, E)$ and a vertex $v \in V$, we shall use $\delta_H(v)$ to denote the set of incident edges of $v$ in $H$.  The rounding algorithm proceeds in 4 steps as we mentioned.  
			\medskip
			
			\noindent{\bf Step 1: Choosing and Marking Candidate Edges.}\  For every job $j \in J$, we independently run the following procedure. We randomly choose two \emph{candidate groups} $v^1_j, v^2_j \in U$, and call the two edges $v^1_jj$ and $v^2_jj$ \emph{candidate edges}; let $H^\cand$ be the bipartite multi-graph between $U$ and $J$ containing all the candidate edges. We guarantee the following two properties for every $j \in J$:
			\begin{itemize}
				\item $\Pr[v^1_j = u] = \Pr[v^2_j = u] = y_{uj}$, for every $u \in U$. 
				\item For every $i \in M$ such that $i$ does not dominate $j$, we have $\Pr[g(v^1_j) = g(v^2_j) = i] = 0$.
			\end{itemize}
			The second property says that we should avoid choosing two candidate groups belonging to the same machine, unless it is unavoidable. 
			\begin{observation}
				If $g(u)$ does not dominate $j$, then $uj$ appears in $H^\cand$ at most once and $\Pr[uj \in H^\cand] = 2y_{uj}$.
			\end{observation}
			There are no other properties needed on how we choose the two candidate edges for a fixed $j$. \medskip
			
			For every $uj \in H^\cand$ such that $g(u)$ does not dominate $j$, we independently \emph{mark} it with probability $\frac{1 - e^{-2y_{uj}}}{2y_{uj}} \leq 1$.  The probability tends to 1 as $y_{uj}$ tends to $0$.  Similarly, we create a bipartite graphs $H^\marked$ between $U$ and $J$ that contains all the marked edges.

			\begin{observation}
				\label{obs:mark-probability}
				If $g(u)$ dominates $j$ for some $u \in U, j \in J$, then $uj \notin H^\marked$. Otherwise, $\Pr[uj \in H^\marked] = 1 - e^{-2y_{uj}}$. Moreover, for a group $u \in U$, the events $uj \in H^\marked$ over all jobs $j \in J$ are independent. 
			\end{observation}
			
			Notice that $H^\marked$ does not contain parallel edges and $H^\marked \subseteq H^\cand$. This holds as if $H^\cand$ contains two copies of the same edge $uj$, then it must be that $g(u)$ dominates $j$, in which case we should not have marked $uj$. \medskip

			\noindent{\bf Step 2: Pairing Edges to Obtain a Set $H^\tsplit$ of Paths and Cycles.}\  For every vertex $u \in U$, we independently run the following pairing procedure. We pair the edges in $\delta_{H^\marked}(u)$ (instead of $\delta_{H^\cand}(u)$), so that 0 or 1 of the edges in $\delta_{H^\marked}(u)$ are not paired, depending on whether $|\delta_{H^\marked}(u)|$ is even or odd. Among all the valid pairings, we choose one uniformly at random. An edge in $\delta_{H^\cand}(u)$ may not be paired for two reasons: First, it may not be in $\delta_{H^\marked}(u)$. Second it is in $\delta_{H^\marked}(u)$, but $|\delta_{H^\marked}(u)|$ is odd and it is the only unpaired edge in the set. 
			
			Now, we create a graph $H^\tsplit$ from $H^\cand$ by splitting groups as follows. Let $H^\tsplit$ contain only the jobs initially. Then, for every group $u \in U$, and any two paired edges $uj, uj' \in H^\cand$, we add a copy of $u$ to $H^\tsplit$, and two edges connecting the copy to $j$ and $j'$ respectively. For every group $u \in U$ and any unpaired edge $uj \in H^\cand$, we add a copy of $u$ to $H^\tsplit$ and connect the copy to $j$. 
			
			Notice that in $H^\tsplit$, every job has degree exactly $2$, and every group (more precisely, every copy of a group) has degree 1 or 2.  So, each connected component in $H^\tsplit$ is either a cycle or a path connecting two copies of groups.  From now on, we may not strictly distinguish between a group in $U$ and its copies, when doing so does not lead to a confusion.

		\begin{claim}
			\label{claim:different-machines}
			Let $P$ be a component in $H^\tsplit$ containing at least $4$ edges,  and let $j \in P$. Then the two adjacent groups of $j$ in $P$ belong to different machines. 
		\end{claim}
		\begin{proof}
			If the two groups belong to the same machine, then the machine must dominate $j$. In this case the two edges between the two groups and $j$ are not marked and thus can not be paired with other edges.  Then $P$ should be a path of length $2$, a contradiction. \hfill
		\end{proof} \medskip
		
	\noindent{\bf Step 3: Breaking Cycles/Paths in $H^\tsplit$ into Segments.}\  We construct a set of edge-disjoint \emph{segments} from the paths and cycles in $H^\tsplit$, defined as follows.
	\begin{definition}
		Let $P$ be a component in $H^\tsplit$ that is a path, or a cycle of length at least $8$. A \emph{segment} $S$ in $P$ is a sub-path of $P$ of length 4 or 6, starting and ending at groups. Moreover, if $S$ is a path $u$-$j$-$u'$-$j'$-$u''$-$j''$-$u'''$ of length $6$, then $g(u) \neq g(u'')$ and $g(u''') \neq g(u')$. 
		
		The groups in $S$ that are not the starting or ending group of $S$ are called the \emph{centers} of $S$: If $S$ is $u$-$j$-$u'$-$j'$-$u''$, then the center is $u'$, and if $S$ is $u$-$j$-$u'$-$j'$-$u''$-$j''$-$u'''$, then the centers are $u'$ and $u''$. 
		
		If $P$ is a length-4 or 6 cycle in $H^\tsplit$, then we also call $P$ a segment, and all groups in $P$ are centers of $P$. 
	\end{definition}

	Notice that by Claim~\ref{claim:different-machines} and the definition, all centers of a segment $S$ belong to different machines, and a center and a non-center group of $S$ belong to different machines. However, it may be possible that the two non-center groups of $S$ belong to the same machine.
		
		We construct a set $\calS$ of edge-disjoint segments. If $P$ is a cycle of length 4 or 6 in $H^\tsplit$, then $P$ is a segment and we add $P$ to $\calS$.  
		
		Focus on a path component $P$ in $H^\tsplit$, and let it be $u_1$-$j_1$-$u_2$-$j_2$-$\cdots$-$u_{\ell+1}$. Then we add a set of edge-disjoint segments of $P$ to $\calS$ as follows: With probability $1/2$, we add to $\calS$ all the segments of $P$ of the form $u_t$-$j_t$-$u_{t+1}$-$j_{t+1}$-$u_{t+2}$ with $t$ being an odd number.  With the remaining probability $1/2$, we add to $\calS$ all the segments with $t$ being an even number. The following claim is immediate:
		\begin{claim}
		    \label{claim:break-into-segments-path}
			Every internal group $u$ in $P$ is the center of a constructed segment with probability $\frac12$.
		\end{claim}
		
		Now we show how to create segments from a cycle component $C$ in $H^\tsplit$ with length at least $8$, in the proof of the following lemma. As the proof is done by some tedious case-by-case analysis, we defer it to the appendix.
		\begin{restatable}{lemma}{breakintosegments}
			\label{lemma:break-into-segments}
			Let $C$ be a cycle component in $H^\tsplit$ with length at least $8$.  There is an efficient random procedure that outputs a set of edge-disjoint segments of $C$, such that every group $u \in C$ is a center of a segment with probability at least $\frac{4}{9}$.
		\end{restatable}
		We add the segments constructed by the lemma to $\calS$. \medskip
		
	\noindent{\bf Step 4: Rounding Along Segments.}\  

		For every segment $S \in \calS$, we list all the vertices along $S$ in a sequence, which have length $5$ or $7$; if $S$ is a length-4 or length-6 cycle, we break it at an arbitrary group and list all the vertices from the group. With probability $1/2$, let $\sigma(j)$ for every job $j$ in the sequence be the group $u$ before $j$. With the remaining probability $1/2$, let $\sigma(j)$ for every $j$ be the group after $j$ in the sequence.

		There are many other jobs that are not in any segment in $\calS$ and thus are not assigned. For these jobs $j$, we independently let $\sigma(j)$ be one of its two neighbors in $H^\cand$  uniformly at random. $j$ may not be in a segment because the component $P \in H^\tsplit$ containing $j$ falls in one of the following cases.
		\begin{itemize}
			\item $P$ is a path, $j$ is the first or the last job on $P$.
			\item $P$ is a cycle of length at least $8$, and the segments from $P$ do not cover $j$. This is possible as the segments constructed by the procedure in Lemma~\ref{lemma:break-into-segments} may not cover all jobs.
		\end{itemize}
		Notice that $P$ can not be a cycle of length $2$ as we do not mark an edge $uj$ if $u$ dominates $j$. This finishes the description of the rounding algorithm for the SNC scheme. 
				
	\subsection{The Analysis of Rounding Algorithm for SNC Scheme}

		Now we analyze the algorithm in Section~\ref{subsec:SNC-algorithm} and finish the proof of Theorem~\ref{thm:SN}.  We let $c_j = 1$ if we assign $j$ to $v^1_j$ and $c_j = 2$ if we assign $j$ to $v^2_j$.\footnote{It is possible that $v^1_j = v^2_j$. In this case, we can still tell whether $c_j = 1$ or $c_j = 2$. Indeed the purpose of introducing $c_j$ is to avoid the confusion caused by this case.} 
		To see Property~\ref{property:SN-marginal}, notice that $\Pr[v^1_j = u] = \Pr[v^2_j = u] = y_{uj}$ for every $uj \in U \times J$.  Conditioned on any values of $v^1_j$ and $v^2_j$, we have $\Pr[c_j = 1] = \Pr[c_j = 2] = \frac12$. Therefore, we have $\Pr[v^1_j = u, c_j = 1] = \Pr[v^2_j = u, c_j = 2] = \frac{y_{uj}}{2}$. So, $\Pr[\sigma(j) = u] = y_{uj}$.
		
		To prove Property~\ref{property:SN-independence}, we focus on two vertices $u, u' \in U$ with $g(u) = g(u')$  (it is possible that $u = u'$) and two distinct jobs $j, j' \in J$.  We have $\Pr[v^1_j = u, v^1_{j'} = u'] = \Pr[v^1_j = u]\Pr[v^1_{j'} = u'] = y_{uj}y_{u'j'}$ as the candidate edges for different jobs are chosen independently.  In the rest of this paragraph, we condition on $v^1_j = u, v^1_{j'} = u'$. We show that the probability $c_j = c_{j'} = 1$ is at most $\frac14$.  Notice that the two edges $uj$ and $u'j'$ in $H^\cand$ can not appear in the same cycle of length $4$ or $6$ in $H^\tsplit$, due to $g(u) = g(u')$ and Claim~\ref{claim:different-machines}. If the two edges $uj$ and $u'j'$ in $H^\cand$ do not belong to a same segment in $\calS$, then $j$ and $j'$ are rounded independently in Step 4. In this case, the probability is precisely $\frac14$. Now assume they appear in the same segment $S \in \calS$. As $g(u) = g(u')$, $u$ and $u'$ must be the two end-vertices of $S$, or we have $u = u'$. In either case, the probability $c_j = c_{j'} = 1$ happens with probability $0$.  So, the probability that $c_j = c_{j'} = 1$ is at most $\frac14$.
		
		This proves that $\Pr[v^1_j = u, v^1_{j'} = u', c_j = c_{j'} = 1] \leq \frac{y_{uj}y_{u'j'}}{4}$.  We can prove the same upper bound for $\Pr[v^1_j = u, v^2_{j'} = u', c_j = 1, c_{j'} = 2], \Pr[v^2_j = u, v^1_{j'} = u', c_j = 2, c_{j'} = 1]$ and $\Pr[v^2_j = u, v^2_{j'} = u', c_j = c_{j'} = 2]$. Therefore, we have $\Pr[\sigma(j) = u, \sigma(j') = u'] \leq y_{uj}y_{u'j'}$. 
		
		\paragraph{Proof of Property~\ref{property:SN-SN}} Finally, we prove Property~\ref{property:SN-SN}, which is the hardest. We focus on a vertex $u \in U$, and two distinct jobs $j, j' \in J$ such that $u$ does not dominate either $j$ or $j'$. Notice that each of $uj$ and $uj'$ appears in $H^\cand$ at most once, and $\Pr[uj, uj' \in H^\cand] = 4y_{uj}y_{uj'}$. Under the event, they are both marked with probability $\frac{(1-e^{-2y_{uj}})(1-e^{-2y_{uj'}})}{4y_{uj}y_{uj'}}$. 
  
		From now on, we condition on the event $uj, uj' \in H^\marked$.
		 Let $Q =  |\delta_{H^\marked}(u)| - 2$.  Then the distribution for $Q$ is dominated by $\text{Pois}(2\rho)$, where for convenience we define $\rho:=1 - y_{uj} - y_{uj'}$, and $\text{Pois}(\lambda)$ is the Poisson distribution with mean $\lambda$. Formally, we have the following claim:
		 \begin{claim}
			  We can define a random variable $X \sim \text{Pois}(2\rho)$ that is correlated with $Q$, such that $Q \leq X$ always holds.
		 \end{claim}
		 \begin{proof}
		 	$Q$ is the number of marked edges incident to $u$ other than $uj$ and $uj'$.  By Observation~\ref{obs:mark-probability}, for every $j'' \notin \{j, j'\}$, the probability that $j''$ is marked is at most $1 - e^{-2y_{uj''}}$ and the events are independent. 
		 	The Bernoulli variable with mean $1 - e^{-2y_{uj''}}$ is dominated by the Poisson distribution with mean $2y_{uj''}$. Therefore we can define a random variable $X_{j''}$ for every job $j'' \notin \{j, j'\}$, such that $X_{j''} \sim \text{Pois}(2y_{uj''})$, $uj'' \in H^\marked$ implies $X_{j''} \geq 1$, and the variables $X_{j''}$ over all $j''$ are independent.  It always holds that $Q \leq \sum_{j'' \in J \setminus \{j, j'\}}X_{j''}$, which follows the Poisson distribution with mean $\sum_{j'' \in J \setminus \{j, j'\}}2y_{uj''} = 2(y(u, J) - y_{uj} - y_{uj'}) \leq 2(1-y_{uj} - y_{uj'}) = 2\rho$ as $y(u, J) \leq 1$. \hfill
		 \end{proof}

		 Under the value $Q$, the probability that $uj$ and $uj'$ are paired is $\frac{1}{Q+1}$ when $Q$ is even, and $\frac{1}{Q+2}$ when $Q$ is odd.  So conditioned on $uj, uj' \in H^\marked$, we have that $uj$ and $uj'$ are paired with each other with probability at least 
		\begin{align*}
			&e^{-{2\rho}}\left(\sum_{t \geq 0\text{ even}}\frac{{(2\rho)}^t}{t!}\frac{1}{t+1} + \sum_{t \geq 1\text{ odd}}\frac{{(2\rho)}^t}{t!}\frac{1}{t+2} \right) \\
			&=  e^{-{2\rho}}\left(\sum_{t = 0}^\infty\frac{{(2\rho)}^t}{t!}\frac{1}{t+1} - \sum_{t \geq 1\text{ odd}}\frac{{(2\rho)}^t}{t!}\frac{1}{(t+1)(t+2)} \right)\\
			&= e^{-{2\rho}}\left(\frac1{2\rho}\sum_{t=1}^\infty\frac{{(2\rho)}^t}{t!} - \frac{1}{{(2\rho)}^2}\sum_{t \geq 3\text{ odd}}\frac{{(2\rho)}^t}{t!}\right)\\
			& = e^{-{2\rho}}\left(\frac{1}{{2\rho}}(e^{2\rho} - 1) - \frac{1}{{(2\rho)}^2}(\frac{e^{2\rho}-e^{-{2\rho}}}{2}  - {2\rho})\right) 
			= \frac1{2\rho} -\frac{1-e^{-{4\rho}}}{8\rho^2},
		\end{align*}
		where we used $\frac{e^{2\rho} - e^{-{2\rho}}}{2} = \sum_{t \geq 1\text{ odd}} \frac{{(2\rho)}^t}{t!}$ in the first equality in the last line.
		
		Define $\calE$ to be the event that $uj$ and $uj'$ are paired to each other, and the copy of the group $u$ for the pair $(uj, uj')$ is a center of some segment in $\calS$.  Notice that when $\calE$ happens, we will not assign both $j$ and $j'$ to $u$.  So we have
		\begin{align}
			\Pr[\sigma(j) = u, \sigma(j') = u] &= 4y_{uj}y_{uj'} \cdot \big(1 - \Pr[\calE|uj,uj' \in H^\cand]\big) \cdot \frac14 \nonumber\\
			&\leq  4y_{uj}y_{uj'} \cdot \left( 1 - \frac{(1-e^{-2y_{uj}})(1-e^{-2y_{uj'}})}{4y_{uj}y_{uj'}} \cdot \left(\frac{1}{{2\rho}} - \frac{1-e^{-{4\rho}}}{8\rho^2}\right) \cdot \frac49 \right) \cdot \frac14 \nonumber\\
			&=y_{uj}y_{uj'} - \frac{1}{9}(1-e^{-2y_{uj}})(1-e^{-2y_{uj'}})\left(\frac{1}{{2\rho}} - \frac{1-e^{-{4\rho}}}{8\rho^2}\right). \label{equ:bound-assign-to-same-group}
		\end{align}
		
		In the inequality, ${(1-e^{-2y_{uj}})(1-e^{-2y_{uj'}})}{4y_{uj}y_{uj'}}$ is the probability that $uj$ and $uj'$ are marked, under the condition that they are candidate edges; $\frac{1}{{2\rho}} - \frac{1-e^{-{4\rho}}}{8\rho^2}$ is the probability that they are paired with each other, under the condition that they are marked; and $\frac49$ comes from Claim~\ref{claim:break-into-segments-path} and Lemma~\ref{lemma:break-into-segments}. 
		
		We then lower bound $\frac{(1-e^{-2y_{uj}})(1-e^{-2y_{uj'}})}{y_{uj}y_{uj'}}\left(\frac{1}{{2\rho}} - \frac{1-e^{-{4\rho}}}{8\rho^2}\right)$.  For a fixed $y_{uj} + y_{uj'} = 1 -\rho$, we have that  $\frac{(1-e^{-2y_{uj}})(1-e^{-2y_{uj'}})}{y_{uj}y_{uj'}}$ is minimized when $y_{uj} = y_{uj'}$.  This follows from taking the logarithm of the quantity, and the fact that $\ln\frac{1-e^{-x}}{x}$ is a convex function of $x$:
		
		\begin{restatable}{claim}{functionconvex}
			$\ln\frac{1-e^{-x}}{x}$ is convex over the domain $[0, \infty)$. 
		\end{restatable}
					
		So for a fixed $\rho \in [0, 1]$, $\frac19\frac{(1-e^{-2y_{uj}})(1-e^{-2y_{uj'}})}{y_{uj}y_{uj'}}\left(\frac{1}{{2\rho}} - \frac{1-e^{-{4\rho}}}{8\rho^2}\right)$ is minimized when $y_{uj} = y_{uj'} = \frac{1 - \rho}{2}$. In this case, the quantity is equal to
		$\displaystyle \frac49\cdot \left(\frac{1-e^{\rho - 1}}{1-\rho}\right)^2 \cdot \left(\frac{1}{{2\rho}} - \frac{1-e^{-{4\rho}}}{8\rho^2}\right)$.
		
		The function in the domain $[0, 1]$ has the minimum value of at least $0.1561$ (See Figure~\ref{fig:rho-graph} in Appendix~\ref{appendix:wC-plotting} for the function). Therefore, \eqref{equ:bound-assign-to-same-group} is at most $(1 - 0.1561)y_{uj}y_{uj'}  = (1 - \eta)y_{uj}y_{uj'}$, finishing the proof of Property~\ref{property:SN-SN}.

	\section{Scheduling Algorithm Based on Rectangle LP}
	\label{sec:wC-algorithm}
	In this section, we describe our algorithm for the problem of scheduling unrelated jobs to minimize the weighted completion time. The problem is formally defined as follows. We are given a set $M$ of machines, a set $J$ of jobs, and a processing time $p_{ij} \in \Z_{>0} \cup \{\infty\}$ for every $i \in M, j \in J$. Every job has a weight $w_j \in \Z_{>0}$ for every $j \in J$. The output of the problem is an assignment $\sigma: J \to M$ and a total order $\prec_i$ of jobs in $\sigma^{-1}(i)$ for every $i \in M$. The completion time $C_j$ of a job $j$ with $\sigma(j) = i$ in the solution is $p_{ij} + \sum_{j' \in \sigma^{-1}(i): j' \prec_i j} p_{ij'}$. The goal of the problem is to minimize $\sum_{j \in J} w_jC_j$. It is well-known that once the assignment $\sigma:J \to M$ is decided, it is optimum to sort jobs in $\sigma^{-1}(i)$ for every $i \in M$ in non-decreasing order of $p_{ij}/w_j$.

	As in \cite{Li20} and \cite{IS20}, our algorithm is based on the so called rectangle LP, which we describe now. Let $T= \sum_{j \in J} \max_{i: p_{ij} \neq \infty} p_{ij}$ so that any schedule will complete by time $T$. The LP is defined as follows: 
	\begin{align}
		\min \qquad \sum_{j \in J}w_j \sum_{i \in M, s \in [0, T)} x_{ijs}(s+p_{ij}) \label{LP:rectangle}
	\end{align}\vspace*{-15pt}
	
	\noindent\begin{minipage}{0.53\textwidth}
		\begin{align}
			\sum_{i \in M, s\in [0, T)}x_{ijs} &= 1 &\quad &\forall j \in J \label{LPC:rectangle-scheduled}\\
			\sum_{j \in J, s \in [t - p_{ij}, t)}x_{ijs} &\leq 1 &\quad &\forall i \in M, t \in [T] \label{LPC:rectangle-capacity}
		\end{align}
	\end{minipage}\hfill
	\begin{minipage}{0.45\textwidth}
		\begin{align}
			x_{ijs} &= 0 &\quad &\forall i \in M, j \in J, s > T - p_{ij} \label{LPC:rectangle-no-late} \\[15pt]
			x_{ijs} &\geq 0	&\quad &\forall i \in M, j \in J, s \in [0, T) \label{LPC:rectangle-non-negative} \\[3pt]\nonumber
		\end{align}
	\end{minipage}\medskip

	In the correspondent integer program, $x_{ijs}$ for every $ij \in M \times J$ and integer $s \in [0, T)$ indicates if job $j$ is scheduled on machine $i$, with starting time $s$.  The objective gives the weighted completion time of the schedule.  \eqref{LPC:rectangle-scheduled} requires that every job $j$ is scheduled. \eqref{LPC:rectangle-capacity} requires that at any time on machine $i$, at most one job is being processed. \eqref{LPC:rectangle-no-late} ensures that no jobs complete after time $T$. \eqref{LPC:rectangle-non-negative} is the non-negativity constraint. The LP is polynomial-sized when $T$ is polynomially bounded, but it can be solved to an arbitrary precision even when $T$ is super-polynomial in $n$, as shown in \cite{Li20}. 
	
	Im and Shadloo \cite{IS20} showed that given a solution $x$ to LP\eqref{LP:rectangle}, one can round it to an integral schedule, whose weighted completion time in expectation is at most $1.488$ times the value of $x$. In this paper, we improve the ratio to $1.45$.
	
	As in \cite{IS20}, we define $R_{ijs}$ for every variable $x_{ijs} > 0$ to be a rectangle with horizontal span $(s, s + p_{ij}]$ and height $x_{ijs}$. We say that the rectangle is for job $j$, on machine $i$ and with starting time $s$.  So \eqref{LPC:rectangle-scheduled} says the total height of all rectangles for a job $j$ is at least 1, and \eqref{LPC:rectangle-capacity} says the total height of rectangles on any machine $i$ that covers a time point $t$ is at most 1.  Let $x_{ij} = \sum_s x_{ijs}$ be the total height of all rectangles for job $j$ on machine $i$, for every $i \in M, j \in J$.

	\subsection{Comparison of Algorithm to That of Im and Shadloo \cite{IS20}} 
	As showed by Li \cite{Li20}, the following independent rounding algorithm leads to a 1.5-approximation for the problem. For every job $j$, we independently choose an ``anchor rectangle'' $R_{ijs}$ for $j$, so that $R_{ijs}$ is chosen with probability $x_{ijs}$. If the anchor rectangle for $j$ is $R_{ijs}$, we assign $j$ to machine $i$, and define $\theta_j$ to be a random real number in the span $(s, s + p_{ij}]$ of the rectangle. Then we schedule all the jobs assigned to a machine $i$ in increasing order of $\theta_j$ values.  This serves as a starting point of their better-than-1.5 approximation algorithms. 
	
	Similar to Im and Shadloo \cite{IS20}, we break the time horizon into a set of base windows of geometrically increasing length and define when a rectangle belongs to a base window. We choose a shifting parameter $\tau_{ij}$ for every pair $ij$ with $x_{ij} > 0$. Roughly speaking, $R_{ijs}$ belongs to base window $k$ if $s + \tau_{ij}$ is in the base window. We create one group for every machine $i$ and every base window $k$, that contains all the rectangles on machine $i$ belonging to the base window $k$.  We guarantee that the total height of rectangles in a group is at most 1. This way, bad rectangles on machine $i$ with the similar width have a good probability of being grouped together. Then applying Theorem~\ref{thm:SN} with this grouping gives some negative correlation between the events these bad rectangles are chosen as the anchor rectangles. This leads to the better-than-1.5-approximation ratio of the algorithm.

	In our SNC scheme, we do not require that a job is only incident to one group belonging to a machine. So we do not need to choose a representative rectangle $R_{ijs}$ for a $ij$ pair as in Im and Shadloo \cite{IS20}. As a result, our construction automatically satisfies that the total height of rectangles belonging to one group is at most $1$. In contrast, the condition is only satisfied in expectation in \cite{IS20}. Thus, they need to filter $1/2$ fraction of the rectangles in a group and apply Chernoff bound, which only guarantees that the condition holds with a probability smaller than 1. 
	
	\subsection{Formal Description of Rounding Algorithm for Scheduling Problem} 
	In this section, we formally describe the rounding algorithm for the problem $R||\sum_j w_jC_j$.\medskip
	
	\noindent{\bf Step 1: Constructing Base Windows and Choosing Shifting Parameters.}\  Let $\alpha  = 0.3, \beta = 12.1$ be two parameters whose values are chosen to optimize the final approximation ratio; recall that $\eta = 0.1561$ is the parameter in Theorem~\ref{thm:SN}.  We randomly choose a parameter $\rho \in [1, 1+\beta)$ so that $\ln \rho$ is uniformly distributed in $[0, \ln (1+\beta))$.  Then, define a grid point to be a real number of the form $\rho(1+\beta)^k, k \in \Z$.  So the time horizon $(0, \infty)$ is broken into many \emph{base windows} by the grid points: the base window $k$ is the interval $(\rho (1+\beta)^{k-1}, \rho(1+\beta)^k]$. 
		
	For every pair $ij$ with $p_{ij} \neq \infty$, we choose a shifting parameter $\tau_{ij}$, independently and uniformly at random from $[0, p_{ij})$. Let $\vec \tau$ be the vector $(\tau_{ij})_{i\in M, j \in J:p_{ij} \neq \infty}$.  Till the end of this section, we condition on the values of $\rho$ and $\vec \tau$ when dealing with probabilistic events. 
	
	\begin{definition} 
		We say a rectangle $R_{ijs}$ belongs to the base window $k$ if $s \leq \rho (1+\beta)^{k-1} < s + \tau_{ij} \leq \rho(1+\beta)^k$.  Let $\calR_k$ be the set of all rectangles belonging to the base window $k$. We use $R_{ijs} \sim_k R_{ij's'}$ to denote the event that the two rectangles $R_{ijs}$ and $R_{ij's'}$ on the same machine $i$ belong to the same base window $k$. We use $R_{ijs} \sim R_{ij's'}$ to denote the event that there exists some $k$ such that $R_{ijs} \sim_k R_{ij's'}$. 
	\end{definition}
	It is possible that a rectangle $R_{ijs}$ does not belong to any base window: This happens if there are no grid points between $s$ and $s + \tau_{ij}$. We note that the requirement $s \leq \rho(1+\beta)^{k-1}$ will be used in the proof of Lemma~\ref{lemma:total-height-small}. \medskip
	
	\noindent{\bf Step 2: Defining Groups and Applying Theorem~\ref{thm:SN}.}\  To apply the SNC scheme, we need to define the groups $U$, the mapping $g: U \to M$ and the fractional matching $y \in [0, 1]^{U \times J}$.  We view each group $u$ as a set of rectangles and the set $U$ of all groups forms a partition of the rectangles.  
	
	For every machine $i$ and base window $k$, we add a group $u_{ik} := \{R_{ijs} \in \calR_k\}$ to $U$, if the set is not empty. We define $g(u_{ik}) = i$. 	For any $ij$ pair in $M \times J$, we add a group $v_{ij}:=\{R_{ijs} \notin \union_k \calR_k\}$ to $U$ if the set is not empty. We define $g(v_{ij}) = i$. That is, $v_{ij}$ contains all the rectangles for $j$ on $i$ that do not belong to any base window.  Clearly, $U$ forms a partition of all rectangles.  Then for every group $u \in U$ with $g(u) = i$ and $j \in J$, we define
	\begin{align*}
	    y_{uj} = \sum_{s: R_{ijs} \in u} x_{ijs}
	\end{align*}
	to be the total height of rectangles in group $u$ for job $j$. By definition, we have $y_{v_{ij}j'} = 0$ if $j \neq j'$.
	
	\begin{lemma}
	    \label{lemma:total-height-small}
		For every $u \in U$ we have $y(u, J) \leq 1$. For every $j \in J$, we have $y(U, j) = 1$.
	\end{lemma}
	\begin{proof}
		For a group $u_{ik} \in U$, we have $y(u_{ik}, J)$ is the total height of all rectangles in $u_{ik}$. 
  Notice that $R_{ijs} \in \calR_k$ only if $[s, s + p_j)$ contains $\rho(1+\beta)^{k-1}$.    
  By \eqref{LPC:rectangle-capacity}, this is at most $1$. $y(v_{ij}, J) = y_{v_{ij}j} \leq 1$ as the total height of rectangles for $j$ is at most 1.  The second statement follows from that $U$ is a partition of all rectangles.	
	\end{proof}	
	
	Now we can apply the algorithm for the SNC scheme in Theorem~\ref{thm:SN} on $U, g$ and $y$ to obtain an assignment $\sigma: J \to U$ of jobs to groups.  For every $j \in J$, we identify an \emph{anchor} rectangle $R_{ijs}$ that resulted in the assignment of $j$ to $\sigma(j)$ in the SNC scheme: assuming $\sigma(j) = u$, we choose the anchor rectangle for $j$ randomly from the all rectangles for $j$ in $u$, with probability proportional to their heights. So, such a rectangle $R_{ijs}$ is chosen with probability $\frac{x_{ijs}}{y_{uj}}$.
	
	We use $i_j$ and $s_j$ to denote the machine and starting time of the anchor rectangle for job $j$.	Abusing notations slightly, we also use $R_{ijs}$ to denote the event that $R_{ijs}$ is the anchor rectangle for $j$.  That is, $R_{ijs}$ holds iff $i_j = i$ and $s_j = s$. We state a lemma that is a direct consequence of Theorem~\ref{thm:SN}. Recall that a machine $i$ dominates a job $j$ if $y(g^{-1}(i), j) = x_{ij} > \frac12$. 
	
	\begin{restatable}{lemma}{applyingSN}
		\label{lemma:applying-SN}
	Let $j, j' \in J$ be two distinct jobs, $i \in M$ and $s, s' \in \Z_{\geq 0}$. Then the following statements hold.
		\begin{enumerate}[label=(\ref{lemma:applying-SN}\alph*)]
			\item \label{property:applying-SN-marginal} $\Pr[R_{ijs}] = x_{ijs}$.
			\item \label{property:applying-SN-independence} $\Pr[R_{ijs}, R_{ij's'}] \leq x_{ijs}x_{ij's'}$.  
			\item \label{property:applying-SN-SN} If $i$ does not dominate any of $j$ and $j'$, then
				\begin{align*}
					\Pr[R_{ijs}, R_{ij's'} |R_{ijs} \sim R_{ij's'}] &\leq (1 - \eta)x_{ijs}x_{ij's'},\text{ where $\eta = 0.1561$ is as in Theorem~\ref{thm:SN}}.
				\end{align*}
		\end{enumerate}
	\end{restatable}

	\noindent{\bf Step 3: Construction of Final Schedule.}\  
	We define 
	\begin{align*}
		\theta_j= \begin{cases}
			(1+\alpha)s_j + \tau_{i_jj} & \text{if $i_j$ does not dominate $j$}\\
			(1+\alpha)s_j + \tau_{i_jj} + 0.2\cdot p_{i_jj} & \text{if $i_j$ dominates $j$}
		\end{cases}
	\end{align*}
 
	Recall that $\alpha = 0.3$ is the global parameter we chose at the beginning of the algorithm. As observed in the algorithms of \cite{Li20} and \cite{IS20}, the hard cases come from the rectangles $R_{ijs}$ with $s \ll p_{ij}$ and $x_{ij}$ close to 0. This is because we can charge some delay to $s$ when it is big and a self-``charging'' is possible when $x_{ij}$ is large.
 Thus we can afford to add the term $\alpha s_i$ to the definition of $\theta_j$; when $x_{i_jj} > \frac12$, we can even afford to add the term $0.2\cdot p_{i_jj}$.
	
	In the final schedule, every job $j$ is scheduled on the machine $i_j$. Jobs assigned to a same machine $i$ are scheduled in increasing order of their $\theta_j$ values (with probability 1, no two jobs have the same $\theta_j$ values). Notice that once the assignment of jobs to machines is fixed, it is the best to sort jobs using the Smith ratios on each machine. However, as in \cite{Li20} and \cite{IS20}, it is more convenient to analyze this (possibly) sub-optimal schedule.   This finishes the description of the algorithm.

\section{Analysis of Rounding Algorithm for Scheduling Problem}
    \label{sec:wC-analysis}
	In this section, we analyze the approximation ratio achieved by the rounding algorithm in Section~\ref{sec:wC-algorithm}. We fix a job $j^*$ and let $C^*$ be the completion time of $j^*$ in the schedule produced by the rounding algorithm. Most of the time we also fix a machine $i$ and our goal is to prove
	\begin{align}
		\E[C^* | i_{j^*}  = i] \leq \frac{1.45}{x_{ij^*}}\sum_{s^*} x_{ij^*s^*}(s^* + p_{ij^*}). \label{inequ:bound-C*-i}
	\end{align}
	This shall prove that $\E[C^*] \leq 1.45\sum_{i,s^*} x_{ij^*s^*}(s^* + p_{ij^*})$, which in turn proves that the expected weighted completion time of the schedule is at most $1.45$ times the value of the rectangle LP. 
	
	As in \cite{IS20}, we can bound the ratio in \eqref{inequ:bound-C*-i} by $1.5 - \frac{1/2}{2} + 0.2 =  1.45$, for the  case when $i$ dominates $j^*$. The $-\frac{1/2}{2}$ term comes from that $x_{ij^*} > 1/2$ and the $0.2$ term comes the additive term $0.2\cdot p_{i_{j^*}j^*}$ in the definition of $\theta_{j^*}$.  We defer the proof of following lemma to the appendix, as it follows from the argument in \cite{IS20}. 
	\begin{restatable}{lemma}{casedominates}
		\label{lemma:case-dominiates}
		If the machine $i$ dominates $j^*$, then \eqref{inequ:bound-C*-i} holds.
	\end{restatable}
	
	\subsection{When $i$ Does Not Dominate $j^*$}
	From now on we consider the case where $i$ does not dominate $j^*$. We fix $j^*$, $i$ and a starting time $s^*$. We shall upper bound $\E[C^* | R_{ij^*s^*}]$ by $1.45(s^* + p_{ij^*})$. De-conditioning on $s^*$ proves \eqref{inequ:bound-C*-i}. Throughout this section we use $\tau^*$  and $\theta^*$ as shorthands for $\tau_{ij^*}$ and $\theta_{j^*}$ respectively. Notice that $\theta^* = (1 + \alpha)s^* + \tau^*$ conditioned on $R_{ij^*s^*}$.  Most of the time we also condition on the value $\tau^*$. So for convenience, we use $\widehat \Pr[\cdot]$ and $\widehat \E[\cdot]$ to denote $\Pr[\cdot|R_{ij^*s^*}, \tau^*]$ and $\E[\cdot|R_{ij^*s^*}, \tau^*]$ respectively. 
	
	As in \cite{Li20} and \cite{IS20}, we define configurations as follows:
	\begin{definition}
		A configuration $f$ is a set of rectangles on $i$ for jobs in $J \setminus \{j^*\}$ with disjoint spans. That is, for every two distinct rectangles $R_{ijs}, R_{ij's'} \in f$, we have $(s, s+p_{ij}]$ and $(s', s' + p_{ij'}]$ are disjoint. 
	\end{definition}
	
	It is well-known that we can find a family $\calF$ of configurations\footnote{Here, a configuration implies a complete schedule on a machine, i.e. an ordered subset of jobs. Thus, each job in a configuration $f$ can be thought of as a rectangle of height $z_f$.}, each $f \in \calF$ with a value $z_f \in [0, 1]$ such that $\displaystyle \sum_{f \in \calF} z_f = 1$ and $\displaystyle \sum_{f \in \calF: R_{ijs} \in f} z_f = x_{ijs}$ for every rectangle $R_{ijs}$ with $j \neq j^*$.  For every $f \in \calF$, we define 
	\begin{align*}
		\Phi(f) := z_f \sum_{R_{ijs} \in f}\frac{1}{x_{ijs}} \cdot \widehat\Pr[R_{ijs}, \theta_j < \theta^*]\cdot p_{ij}.
	\end{align*}
	
	The usefulness of the definition $\Phi(f)$ can be seen from the following claim:
	\begin{restatable}{claim}{ECbyPhif}
	    \label{claim:bound-E-C*-by-Phi-f}
	    $\displaystyle\widehat \E[C^*] - p_{ij^*} = \sum_{f \in \calF} \Phi(f)$.
	\end{restatable}
	
	As $\sum_{f \in \calF}z_f = 1$, our goal then is to upper bound $\frac{\Phi(f)}{z_f}$ for every $f \in \calF$, considerably below $\theta^*$. This will bound $\hat \E[C^*]$ away from $\theta^* + p_{ij^*} = (1 + \alpha)s^* + \tau^* + p_{ij^*}$. As $\E[\tau^*| i_{j^*} = i] = \frac{p_{ij^*}}{2}$, we shall obtain a better-than-1.5-approximation.  We bound $\frac{\Phi(f)}{z_f}$ in the following crucial lemma. For any $\rho \in [1, 1+\beta)$ and a real $t > 0$, let $h_\rho(t)$ be the smallest grid point that is at least $t$, under the value of $\rho$: $h_\rho(t) = \rho(1+\beta)^k$ for the smallest integer $k$ with $\rho(1+\beta)^k \geq t$.
	\begin{restatable}{lemma}{Phiftozf}
		\label{lemma:bound-Phi-f}
		For every $f \in \calF$, we have 
		\begin{align}
			\frac{\Phi(f)}{z_f} &\leq \theta^* - \min\left\{
				\eta\cdot \E_\rho \int_0^{\theta^*} {\bf1}\left(s^*< \frac{h_\rho(s^* + \tau^*)}{1+\beta} <  \tau < h_\rho(s^* + \tau^*)\right)\cdot \sfd\tau, \quad 0.102\theta^*\right\}. \label{inequ:bound-Phi-f}
		\end{align}
	\end{restatable}
	We defer the proof to the appendix, and elaborate more on the inequality and proof here. For simplicity, assume $i$ does not dominate any job involved in $f$; this is the bottleneck case. We shift a rectangle $R_{ijs}$ to the right by $\alpha s$, so that the value of $\theta_j$ will be the starting time of the anchor rectangle for $j$ after shifting, plus $\tau_{ij}$. After shifting, the horizontal span of rectangles in the configuration $f \in \calF$ are disjoint, as we shift the rectangles on the right more than the rectangles on the left.  Then, without using the negative correlation, we can bound $\frac{\Phi(f)}{z_f}$ by $\theta^*$.  The negative term in \eqref{inequ:bound-Phi-f} comes from two cases, corresponding to the two terms in the min operator. The first term, which is the main term, comes from the negative correlation between $R_{ij^*s^*}$ and  the unique rectangle $R_{ijs}$ in $f$ that covers $\theta^*$ after shifting.  We shall show that the worst case happens when $s = 0$. In the integration, $\tau$ is the $\tau_{ij}$ value and the condition inside $\bf1(\cdot)$ is for the event $R_{ijs} \sim R_{ij^*s^*}$. The second term $0.102\theta^*$ in the min operator comes from the case where the rectangle $R_{ijs}$ has $s \geq s^*$.

	\medskip
	
	To avoid the issue of dividing by $0$, we assume $s^* > 0$. One can show that the case where $s^* = 0$ will be covered by the case $s^* > 0$ and $p_{ij^*}$ tends to $\infty$. For every $r \in [0, \infty)$,  we define
	\begin{align}
		Q^\circ(r) &= \eta\cdot\E_\rho \int_{0}^{1+\alpha + r} {\bf1}\Big(1 < \frac{h_\rho(1 + r)}{1+\beta} < y < h_\rho(1 + r)\Big)\cdot \sfd y, \\
		Q(r) &= \min\left\{Q^\circ(r), \quad 0.102(1+\alpha + r)\right\}.  \label{equ:define-Q}
	\end{align}

	\begin{corollary}
		Let $r = \tau^*/s^*$. Then we have $\widehat\E[C^*] - p_{ij^*} \leq \big(1+\alpha + r - Q(r)\big)s^*$.
	\end{corollary}
	\begin{proof}
		First, for every $f \in \calF$, we have $\frac{\Phi(f)}{z_f} \leq \theta^* - Q(r)\cdot s^*$.
				This follows from Lemma~\ref{lemma:bound-Phi-f} by defining $r = \frac{\tau^*}{s^*}$ and using the fact that $\ln \rho$ is uniformly distributed over $[0, \ln(1+\beta))$. 
		By Claim~\ref{claim:bound-E-C*-by-Phi-f}, we have
		\begin{flalign*}
			&& \widehat\E[C^*] - p_{ij^*} &= \sum_{f \in \calF} \Phi(f) \leq \sum_{f \in \calF} (\theta^* - Q(r)s^*)z_f = (1+\alpha)s^* + \tau^* - Q(r)s^* &&\\
			& &&= \big(1 + \alpha + r - Q(r)\big) s^*. && 
		\end{flalign*}\vspace*{-30pt}
		
		\hfill
	\end{proof}

	Now we decondition on $\tau^*$ and bound the final ratio for the case that $i$ does not dominate $j^*$. Let $p^* = p_{ij^*}$, $o = \frac{p^*}{s^*}$, and $r = \frac{\tau^*}{s^*}$.
		\begin{lemma} \label{lemma:bound-C^*-decondition-tau^*}
	$\displaystyle \frac{\E[C^*|R_{ij^*s^*}]}{s^* + p^*} \leq \frac{1}{1+o}\cdot \left(1 + \alpha + \frac{3o}{2} - \frac{1}{o}\int_{0}^o Q(r)\sfd r\right)$.
\end{lemma}\vspace*{-10pt}
\begin{flalign*}
	\textit{Proof.} &&\E[C^*|R_{ij^*s^*}] - p^* &\leq \int_{0}^{p^*} \E\big[C^*|R_{ij^*s^*}, \tau^*\big]\cdot \frac{\sfd  \tau^*}{p^*} - p^* \leq \frac{{s^*}^2}{p^*} \cdot \int_{0}^{p^*/s^*} (1 + \alpha + r - Q(r)) \sfd  r&&\\
	&& &=s^* \cdot\frac1o \int_0^o (1+\alpha + r - Q(r)) \sfd r =s^*\left(1+\alpha + \frac o2 - \frac1o\int_0^o Q(r)\sfd r\right). &&
\end{flalign*}\vspace*{-5pt}
\begin{flalign*}
	\text{Then} && \frac{\E[C^*|R_{ij^*s^*}]}{s^* + p^*} &\leq \frac{1}{(1+o)s^*}\cdot s^*\left(1+\alpha + \frac {3o}2 - \frac1o\int_0^o Q(r)\sfd r\right) &&\\
	&& &=\frac{1}{1+o}\cdot \left(1+\alpha + \frac {3o}2 - \frac1o\int_0^o Q(r)\sfd r\right). &&
\end{flalign*}

It remains to show that $\frac{1}{1+o}\cdot \left(1+\alpha + \frac {3o}2 - \frac1o\int_0^o Q(r)\sfd r\right) \leq 1.45$ for every $o \geq 0$. This will prove $\E[C^*|R_{ij^*s^*}] \leq 1.45(s^* + p^*)$, which proves \eqref{inequ:bound-C*-i} for the case where $i$ does not dominate $j^*$.  The proof of the inequality is by tedious elementary mathematical techniques and so we defer most of the proofs to the appendix. 
	
First, we can obtain the closed form for $Q^\circ(r)$: 
	\begin{restatable}{lemma}{Qcircr}
		\label{lemma:Q-circ-r}
		\begin{align*}
			Q^\circ(r) = \begin{cases}
				\frac{\eta}{\ln(1+\beta)}\big((1 + \alpha + r)\cdot\ln(1+r) - r\big) & r \in [0, \beta - \alpha]\\
				\frac{\eta}{\ln(1+\beta)}\left((1 + \alpha + r) \ln\frac{(1 + \beta)(1 + r)}{1 + \alpha + r} - (\beta - \alpha)\right) & r \in (\beta - \alpha, \beta]\\
				\frac{\eta}{\ln(1+\beta)}\left(\alpha  - \frac{\beta(1 + r)}{1 + \beta} + (1 + \alpha + r)\ln\frac{(1+\beta)(1 + r)}{1 + \alpha + r}\right) & r > \beta
			\end{cases}
		\end{align*}
	\end{restatable}
	
	Figure~\ref{fig:Qr} in Appendix~\ref{appendix:wC-plotting} shows the function $Q^\circ(r)$ over $r \geq 0$, the function $0.102(1+\alpha + r)$ over $r \in [0, \beta - \alpha)$ and the function $0.1 r$ over $r \in [\beta - \alpha, \infty)$. It shows that $Q^\circ(r) \leq 0.102(1+\alpha + r)$ (this is the other term in the definition of $Q(r)$) when $r \in [0, \beta - \alpha]$ and $Q^\circ(r) \geq 0.1 r$ when $r \geq \beta - \alpha$.  For formality, we prove the later inequality analytically as $r$ can go to $\infty$ in this case. 
		\begin{restatable}{lemma}{Qrratio}
		\label{lemma:Q-r-ratio}
		When $r \geq \beta - \alpha$, we have $Q^\circ(r) \geq 0.1r$. 
	\end{restatable}

	We define $Q'(r)$ for every $r \geq 0$ as follows:
	\begin{align*}
		Q'(r) = \begin{cases}
			Q^\circ(r) = \frac{\eta}{\ln(1+\beta)}\big((1+\alpha+r)\cdot\ln(1+r) - r\big) & r \in [0, \beta - \alpha]\\
			0.1r & r  > \beta - \alpha
		\end{cases}
	\end{align*}
	Therefore, we have $Q'(r) \leq Q(r)$.  It remains to prove the following lemma, whose proof is in the appendix.
		\begin{restatable}{lemma}{boundcstar}
			\label{lemma:bound-C^*-decondition}
			We have $\displaystyle\frac{1}{1+o}\cdot \left(1+\alpha + \frac {3o}2 - \frac1o\int_0^o Q'(r)\sfd r\right)\leq 1.45$.
	\end{restatable}

\section{Minimizing $L_k$-norms of Machine Loads}
    \label{sec:lk}

In this section we consider the other objective of minimizing $L_k$-norms of machine loads. As before, we are given as input $\{p_{ij}\}_{j \in J, i \in M}$ where $p_{ij}$ indicates $j$'s processing time (or size) on machine $i$. Formally, for an assignment $\sigma: J \rightarrow M$ of jobs to machines, $\sum_{j \in \sigma^{-1}(i)} p_{ij}$ is machine $i$'s load, i.e., the total processing time of jobs assigned to machine $i$. It is also 
equivalent to the maximum completion time of jobs in $\sigma^{-1}(i)$ when processing them on machine $i$ with no idle times. Neither migration nor preemption is allowed. 
The goal is to find an assignment that minimizes the following objective:
$$ 
\left(\sum_{i \in M}\left(\sum_{j \in \sigma^{-1}(i)} p_{ij}\right)^k \right)^{1/ k}
$$
We will optimize the objective after taking the $k$th power. In other words, we will optimize the $L_k^k$-norm of machine loads and take the $k$th root at the end.

\subsection{Linear Programming}

We again use a time indexed LP to find an approximate solution:

\begin{align}
\min  \sum_{j \in J, i \in M, s \in [T]} w_{ijs}  x_{ijs}  \nonumber\\
	\sum_{i \in M, s \in [T]} x_{ijs}  &\geq 1 &\forall j \in J  \nonumber \\
	\sum_{j \in J, s \in (t - p_{ij}, t]} x_{ijs} &\leq 1 &\forall i \in M, t \in [T] \label{eqn:lk-capacity} \\
	x_{ijs} & = 0 &\forall i \in M,j \in J,s > T - p_{ij} \nonumber \\
	x_{ijs} & \in [0, 1] &\forall i \in M,j \in J,s  \in [T], \nonumber
\end{align}
where $w_{ijs} := (s + p_{ij})^k - s^k$.

 The time indexed LP is almost identical to the LP presented at the beginning of Section~\ref{sec:wC-algorithm}. The main difference lies in the objective where job $j$ adds to the $k$th power of machine loads by $w_{ijs}$ when job $j$ starts execution at time $s$ on machine $i$. To better understand the objective, consider the IP before relaxation; so, $x_{ijs} \in \{0, 1\}$. Fix an assignment of jobs to machines and say jobs $j_1, j_2, \ldots j_H$ are scheduled on machine $i$ in this order and let $C_{j_1}, C_{j_2},  \ldots , C_{j_H}$ be their respective completion time---ordering doesn't affect machine loads so we assume an arbitrary fixed ordering.  Then, machine $i$ contributes to the $L^k_k$-norm objective in the schedule by 
$$\left(C_{j_H}\right)^k = 
\sum_{h' = 1}^H \left((C_{j_{h'}})^k - (C_{j_{h'-1}})^k \right),$$
where $C_{j_0} := 0$ for brevity.
So, when job $j_{h'}$ starts at time $s = C_{j_{h'-1}}$, i.e., $x_{ij_{h'}s} = 1$, we can pretend that it contributes to the objective by $w_{ij_{h'}s} = (s+ p_{ij_{h'}})^k - s^k$. Thus, the above objective captures the $L_k^k$-norm objective of the IP.

All the constraints remain unchanged. The only simplification is that we don't explicitly prohibit job $j$ from being assigned to machines $i$ such that $p_{ij} > T$, but this is only for notational convenience. Thus, the above LP is a valid LP relaxation for minimizing the $L_k^k$-norm objective.

\subsection{Splittable Rectangle Packing} 

To better understand the structure of the LP solution $x$, we consider the following 
fractional packing problem for each machine $i$, which we term \emph{splittable rectangle packing} (\srp). 

Fix a machine $i$. Suppose we have a rectangle of width $x_{ij} := \sum_{s} x_{ijs}$ and height $p_{ij}$ for each job $j$.\footnote{Here, we override the definition of height that was used in the previous sections.} We are allowed to split the rectangle into smaller ones of the \emph{same height}, subject to the constraint that the total width of the resulting rectangles equals the original width, $x_{ij}$. We are asked to pack the rectangles into a container of width 1 with no overlap. Then, we can view the packing as a collection $\cI_i$ of subsets of rectangles of the same width $\lambda(I)$, where $\sum_{I \in \cI_i} \lambda(I) = 1$.
Let $p(I)$ denote the total height of rectangles in $I$. Then the packing's cost is $\sum_{I \in \cI_i} \lambda(I) \cdot  (p(I))^k$. See Figure~\ref{fig:packing} for illustration.

\begin{figure}[t]
\centering
\includegraphics[width = .5\textwidth]{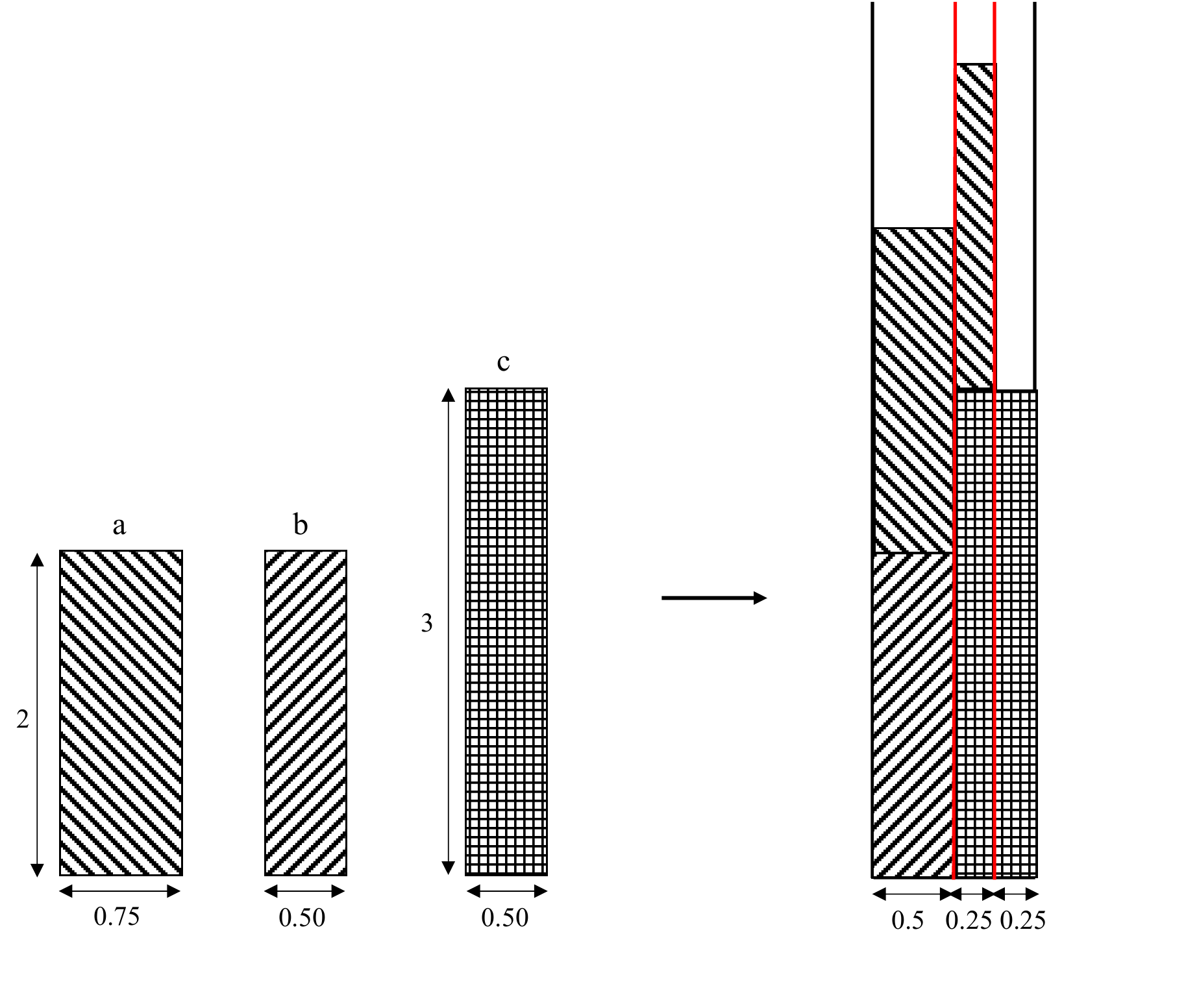}
\caption{Illustration of the splittable rectangle packing. In this example, we are given three rectangles, $a$, $b$, and $c$. They have height $2$, $2$, $3$ and width $0.75, 0.5, 0.5$, respectively. Intuitively they correspond to jobs of processing times $2$, $2$, $3$ that are fractionally assigned to a machine by $0.75, 0.5, 0.5$. The right figure shows a packing of the rectangles. The three groups $I_1 = \{a, b\}$, $I_2 = \{a, c\}$, $I_3 = \{c\}$ have width $0.5, 0.25, 0.25$ respectively. Note that rectangle $a$ was split into two, which appear in $I_1$ and $I_2$. Since $p(I_1) = 2 + 2$, $p(I_2) = 5, p(I_3) = 3$, the objective of the packing is $0.5 * 4^k + 0.25 * 5^k + 0.25 * 3^k$. It is an easy exercise to see that this is indeed an optimum packing for any $k > 1$. Lemma~\ref{lem:packing-equal} shows that this is what the LP optimizes on each machine given a set of jobs fractionally assigned to the machine. 
}
\label{fig:packing}
\end{figure}

Note that $I \in \cI_i$ may include more than one rectangle originating form the same job.\footnote{We can avoid this by using configuration LPs instead of time-indexed LPs although we don't have to.} So, if we associate rectangles with jobs that they originate from, we can view $I$ as a multi-set of jobs and $P(I)$ is the total size of the jobs in the multiset. Let $\srp^*(\{x_{ij}\}_{j})$ denote the optimum solution for  instance $\{x_{ij}\}_{j}$ of the splittable rectangle problem, or its objective. 

If a rectangle is placed at distance $s$ from the bottom of the container, we  say that it starts at time $s$. This is natural because there is a natural mapping from jobs to rectangles as discussed above.

\begin{lemma}
    \label{lem:packing-equal}
	For an optimum solution $x$ to the LP and for each machine $i$, $\srp^*(\{x_{ij}\}_{j})$ is equal to $x$'s contribution to the LP objective on machine $i$, i.e., $\sum_{s, j} w_{ijs} x_{ijs}$. 
\end{lemma}
\begin{proof}
    Fix a machine $i$ and the value of $x_{ij}$ for all $j$. 
    For the sake of contradiction suppose $\srp^*(\{x_{ij}\}_{j}) < \sum_{s, j} w_{ijs} x_{ijs}$. To draw a contradiction, we convert 
    the optimum packing $\cI_i$ to the SRP into a local LP solution on machine $i$. As discussed above, each split rectangle corresponds to a job and its start time is naturally defined in each subset $I$ in $\cI_i$; the rectangles are ordered in an arbitrary but fixed order). Initially $x_{ijs} = 0$ for all $j$ and $s$. For each split rectangle in $I \in \cI_i$ that start at time $s$ and corresponds to $j$, we add $\lambda(I)$ to $x_{ijs}$. Since the rectangles are packed into the container of width 1 with no overlaps, the resulting local LP solution must satisfy constraints~(\ref{eqn:lk-capacity}). This means we can improve the optimum LP solution locally on machine $i$, which is a contradiction.     Therefore, we have
	$\srp^*(\{x_{ij}\}_{j}) \geq \sum_{s, j} w_{ijs} x_{ijs}$.
	
Conversely, we can convert the LP solution on machine $i$ to a valid solution to SRP without changing cost. To this end, create a rectangle $R_{ijs}$ for each $x_{ijs} > 0$. The rectangle has height $p_{ij}$ and width $x_{ijs}$. Sort the rectangles $\{R_{ijs}\}_{j, s}$ in non-decreasing order of their start time and pack them into the container of width 1---here, we split the considered rectangle if necessary. Then, it is easy to see that we can pack $R_{ijs}$ starting it at time $s$ due to the LP constraint (\ref{eqn:lk-capacity}).  As we didn't change the start times, the cost is preserved. Thus, we have $\srp^*(\{x_{ij}\}_{j}) \leq \sum_{s, j} w_{ijs} x_{ijs}$.
\end{proof}

We have shown the equivalence of two views: fractional optimal scheduling in the LP and fractional splittable rectangle packing for fixed $\{x_{ij}\}_j$. Henceforth we will use both views in the analysis. In particular, we will interchangeably use jobs and rectangles.

\subsection{Algorithm}

We use the same algorithm as Kalaitzis \etal \cite{Ola2017unrelated} used for minimizing weighted completion time of jobs with uniform smith ratios, which is essentially the celebrated algorithm developed by Shmoys and Tardos for the Generalized Assignment Problem \cite{shmoys1993approximation}. For completeness we reproduce the algorithm below. 

Sort jobs on each machine $i$ in non-increasing order of their processing time, $p_{ij}$, and fractionally pack them into buckets of unit capacity, denoted as $B(i, 1), B(i, 2), B(i, 3) \ldots$, opening a new bucket only when necessary. Formally, let $X_{ij} := \sum_{j': p_{ij'} \geq p_{ij}} x_{ij'}$ breaking ties in an arbitrary but fixed way. 
If $\lceil X_{ij} - x_{ij} \rceil = \lceil X_{ij} \rceil$, then $j$ is fractionally matched with $B(i, \lceil X_{ij} \rceil)$ by $x_{ij}$ where  $x_{ij} := \sum_s x_{ijs}$ as defined before. Otherwise, job $j$ is fractionally matched with $B(i, \lceil X_{ij} - x_{ij} \rceil)$ by 
$\lceil X_{ij} - x_{ij} \rceil - (X_{ij} - x_{ij})$ and with $B(i, \lceil X_{ij} \rceil)$ by $X_{ij} - \lfloor X_{ij} \rfloor$.\footnote{The only degenerate case is when $X_{ij} - x_{ij}$ and $X_{ij}$ are both integers. In that case, $j$ is fully matched with $B(i, X_{ij})$.} In either case, job $j$ is matched with buckets by $x_{ij}$ in total. If it is matched with more than bucket, it is matched with exactly two consecutive buckets on machine $i$.

For convenience, we add dummy jobs of size 0 (which doesn't affect machine loads), so 
every bucket is fully matched. Thus, we have defined a fractional perfect matching between jobs and buckets. It is well known that a fractional perfect matching can be decomposed into a convex combination of polynomially many integral matchings. We sample an  integral matching at random from the combination/distribution. Clearly, it assigns each job to exactly one bucket, therefore to one machine.

\subsection{Analysis}

Fix a machine $i$. As a result of the above rounding, suppose machine $i$ gets assigned a set $I$ of jobs with probability $\lambda(I)$. Clearly we have $\sum_{I} \lambda(I) = 1$.   Because $j$ is assigned to $i$ with probability $x_{ij}$\footnote{It is an easy exercise to show this as the fractional matching constructed in the algorithm matches $i$ and $j$ by $x_{ij}$ units.}, we can view the rounding outcome as a feasible solution  to the instance $\{x_{ij}\}_j$ of $\srp$, which we denote as $A_i$. Thanks to Lemma~\ref{lem:packing-equal}, our goal becomes upper-bounding this solution's cost relative to $\srp^*(\{x_{ij}\}_{j})$. But the solution can be complicated to analyze due to correlations among jobs that exist on other machines. The key idea in \cite{Ola2017unrelated} lies in  observing a structural property of the rounding scheme's outcome and upper bounding the cost by the worst packing solution satisfying the property.

For brevity, henceforth we fix a machine $i$ and consider the instance $\{x_{ij}\}_{j}$ of SPR. The following definition articulates the key property. Recall that a rectangle's height is equal to the corresponding job' size.  

\begin{definition}
    A solution $\cI$ with $\lambda(I), I \in \cI$ to the SRP is balanced if for any $I, I' \in \cI$ and any integer $q > 0$, the $q$th highest rectangle in $I$ has no smaller height than  the $q+1$th highest rectangle in $I'$. 
\end{definition}

Recall that $A_i$ denotes the solution to the $\srp$ resulting from the rounding; for convenience, we may let it also denote its cost. The following observation is immediate from the algorithm's definition.

\begin{observation}
    It is the case that $A_i$ is balanced. 
\end{observation}

Let $O_i$ denote the optimum solution's cost, i.e., $\sum_{j,s} w_{ijs} x_{ijs}$. 
Knowing that $A_i$ is balanced, we will consider the worst balanced packing solution and compare its cost to $O_i$ for the sake of analysis. Intuitively, the packing gets worse when more bigger jobs are grouped together. 

Formally the worst balanced packing is defined as follows. Sort jobs in decreasing order of $p_{ij}$, breaking ties arbitrarily. If $j_1, j_2, ...,$ is the order, they are mapped to intervals $(0, x_{ij_1}], (x_{ij_1}, x_{ij_1} + x_{ij_2}], \ldots$. Let $f(z)$ denote the size of the job whose interval includes $z$.  See Figure~\ref{fig:f}. Let $F(z) := f(z) + f(z+1) + f(z+2) + \ldots$ and $F_r(z) := F(z) - f(z)$, which are both defined over $z \in [0, 1]$. See Figures~\ref{fig:f-partition} and \ref{fig:bigF}.

% ...
\begin{figure}[H]
\centering
\includegraphics[width = .7\textwidth]{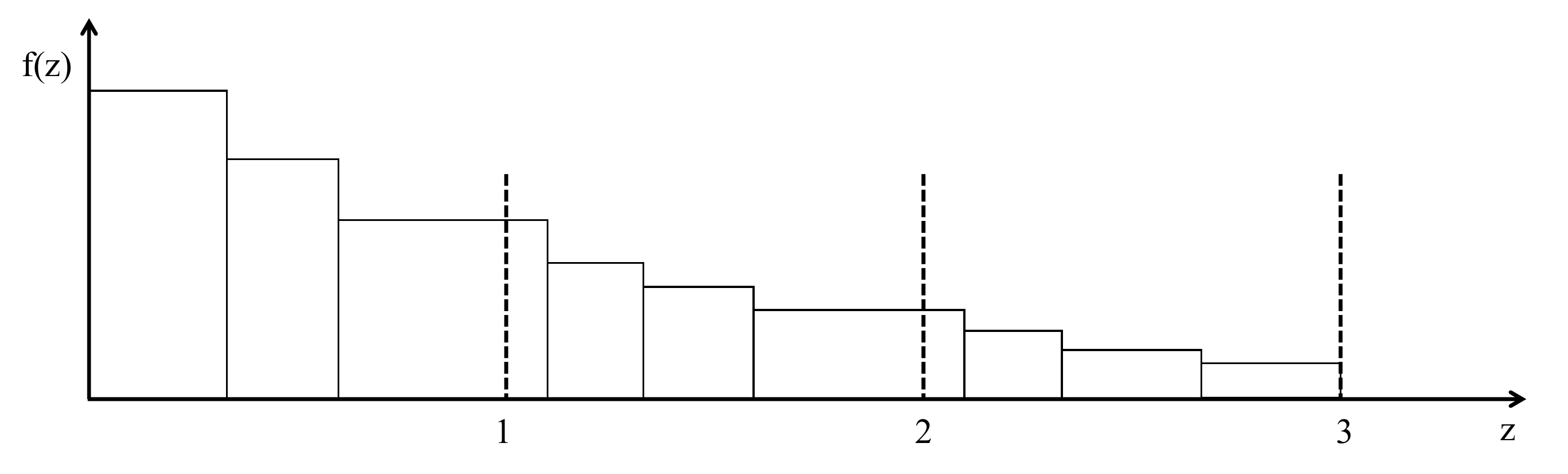}
\caption{Illustration of $f$.}
\label{fig:f}
\end{figure}

\begin{figure}[H]
\centering
\includegraphics[width = .7\textwidth]{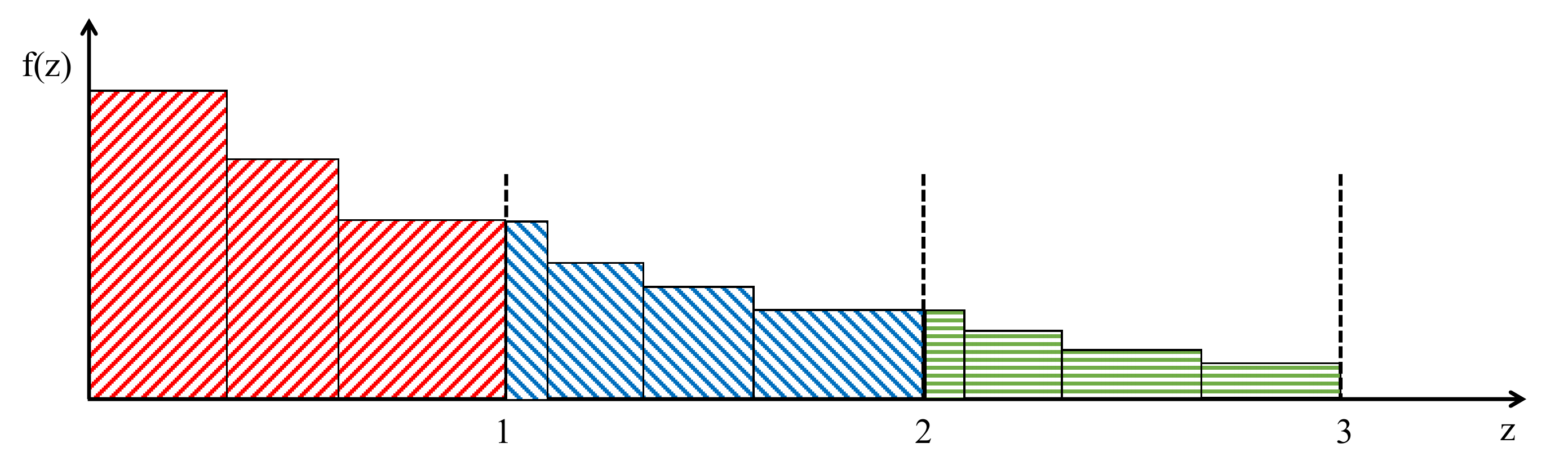}
\caption{Partitioning  $f$.}
\label{fig:f-partition}
\end{figure}

\begin{figure}[H]
\begin{subfigure}{.5\textwidth}
  \centering
  \includegraphics[width=.8\linewidth]{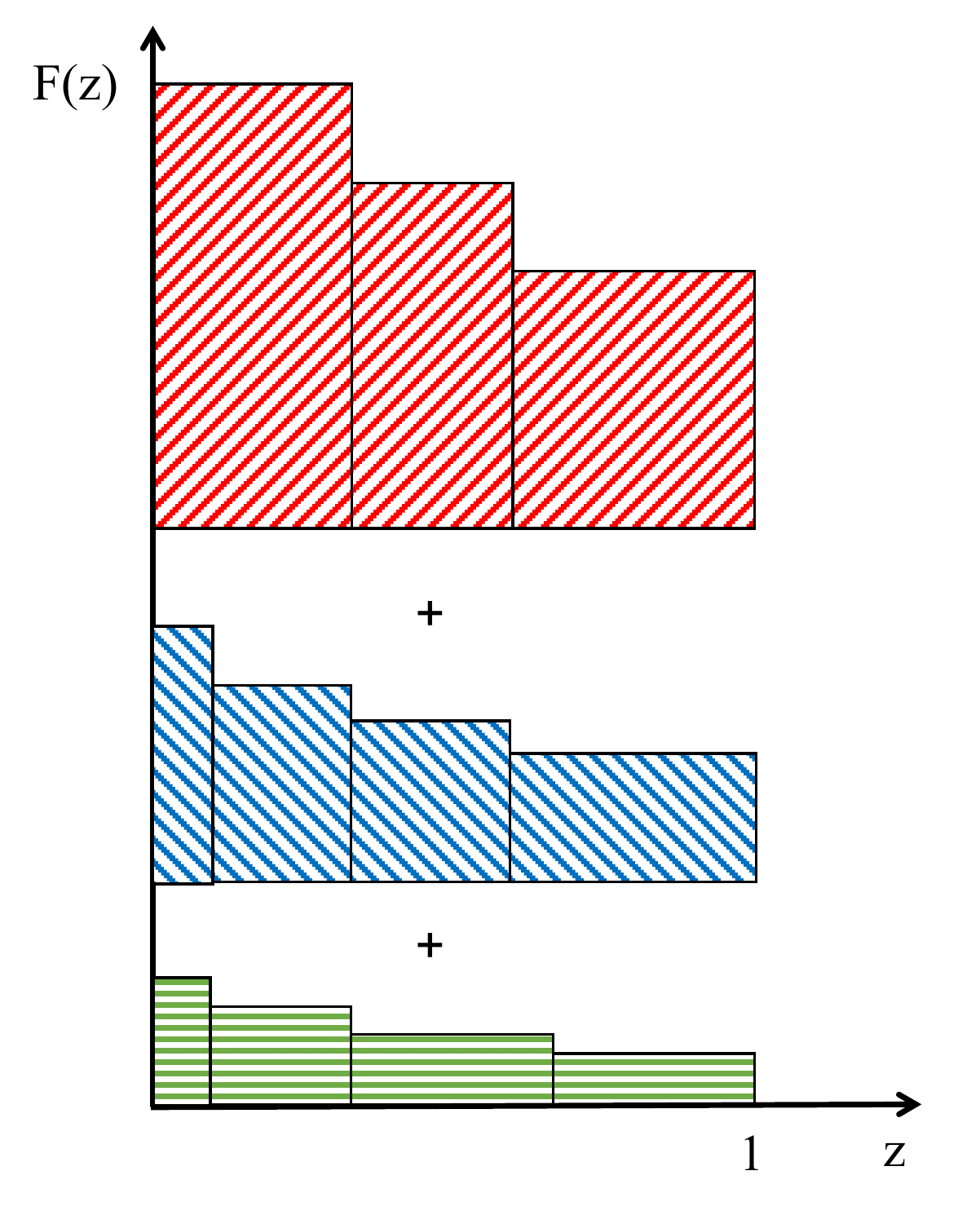}
  \label{fig:sub1}
\end{subfigure}%
\begin{subfigure}{.5\textwidth}
  \centering
  \includegraphics[width=.8\linewidth]{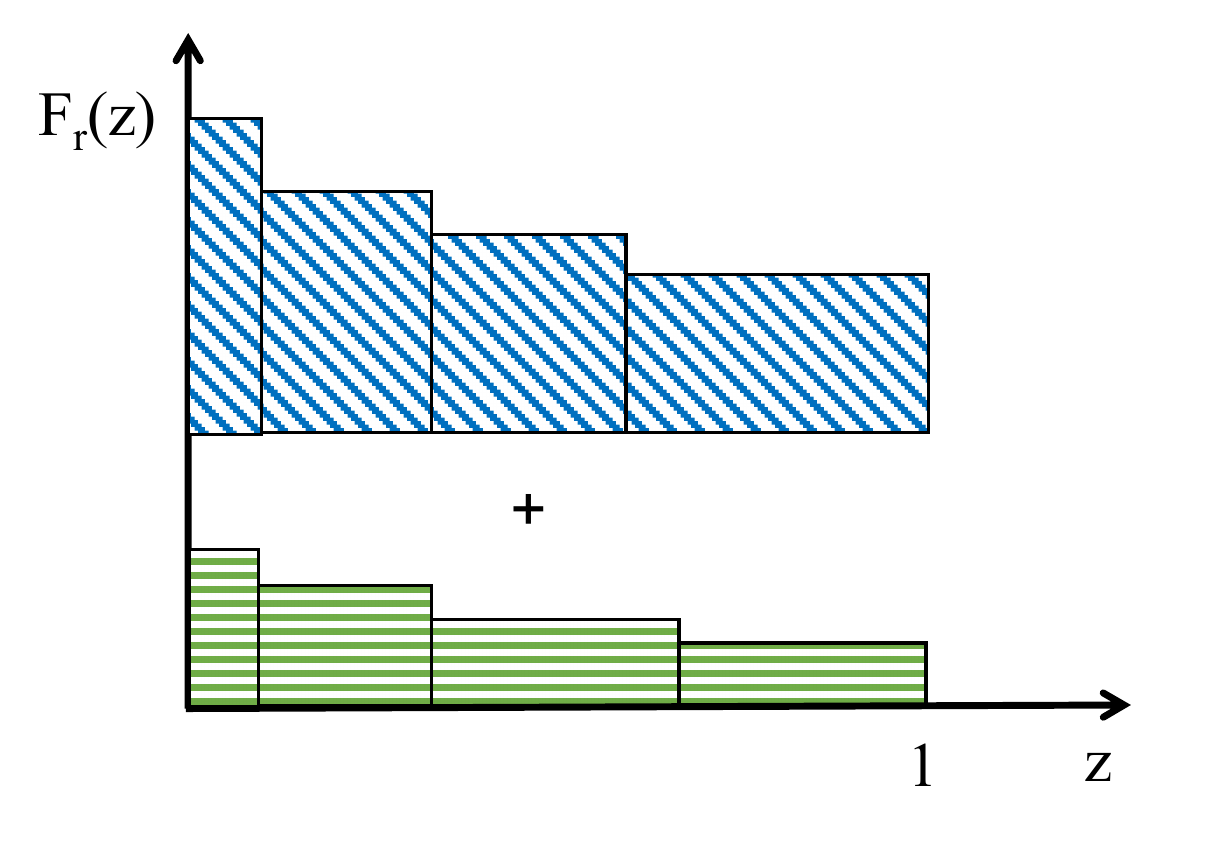}
  \label{fig:sub2}
\end{subfigure}
\caption{Illustration of $F$ and $F_r$.}
\label{fig:bigF}
\end{figure}

\begin{lemma}
	$A_i \leq \int_{0}^1 (F(z))^k d z$.
\end{lemma}
\begin{proof}
The randomized rounding chooses exactly one job from each bucket. In other words, it chooses $z_1 \in [0, 1)$, $z_2 \in [1, 2), \ldots$ at random and the $k$th power of machine $i$'s load is $(f(z_1) + f(z_2) + \ldots)^k$. It is an easy exercise to see that the expected $k$th power is maximized when $z_1 = z$, $z_2 = z+1 \ldots$ for $z$ sampled from $[0, 1)$ uniformly at random. 
In this case the expected $k$th power is exactly the desired upper bound.
The proof follows from a straightforward swapping argument.
\end{proof}	

Having upper bounded $A_i$ by the ``worst" balanced packing's cost, we will compare the worst  balanced packing to the optimum packing. 

\begin{lemma}
    \label{lem:lk-middle-lemma}
    Let $O_i := \sum_{j, s} w_{ijs} x_{ijs}$ and $A'_i := \int_{0}^1 (F(z))^k d z$. If $A'_i \leq \alpha O_i$, then the randomized rounding gives a $\alpha^{1/ k}$-approximation for minimizing $L_k$-norms of machine loads. 
\end{lemma}
\begin{proof}
    We have $\E \sum_i A_i \leq \sum_i A'_i \leq \alpha \sum_i O_i$. Thus, the algorithm yields an assignment that is $\alpha$-approximate against the LP solution for the $L_k^k$-norm objective. By taking the $k$th root, the lemma follows. 
\end{proof}

Therefore, our remaining goal becomes upper bounding $A'_i / O_i$. We will make several modifications of $f$  (and therefore $F$ and $F_r$ accordingly) in a sequence without decreasing $A'_i / O_i$. In each step we convert the current $f$ into $f'$, and accordingly $F$ and $F_r$ into $F'$ and $F'_r$, respectively. For convenience, we will assume that $f$ has domain $[0, \infty)$ by adding dummy jobs of zero size. Note that $F(z)$ and $F_r(z)$ have domain $[0, 1]$.

\paragraph{First step.} $F' = F$; $f'(z) = F(z) - F(1)$, $z \in [0, 1]$; and $F'_r(z) = F(1)$. Further, jobs corresponding to $F'_r(z)$ are infinitesimal. See Figure~\ref{fig:step1}.

\begin{figure}[H]
\centering
\includegraphics[width = .7\textwidth]{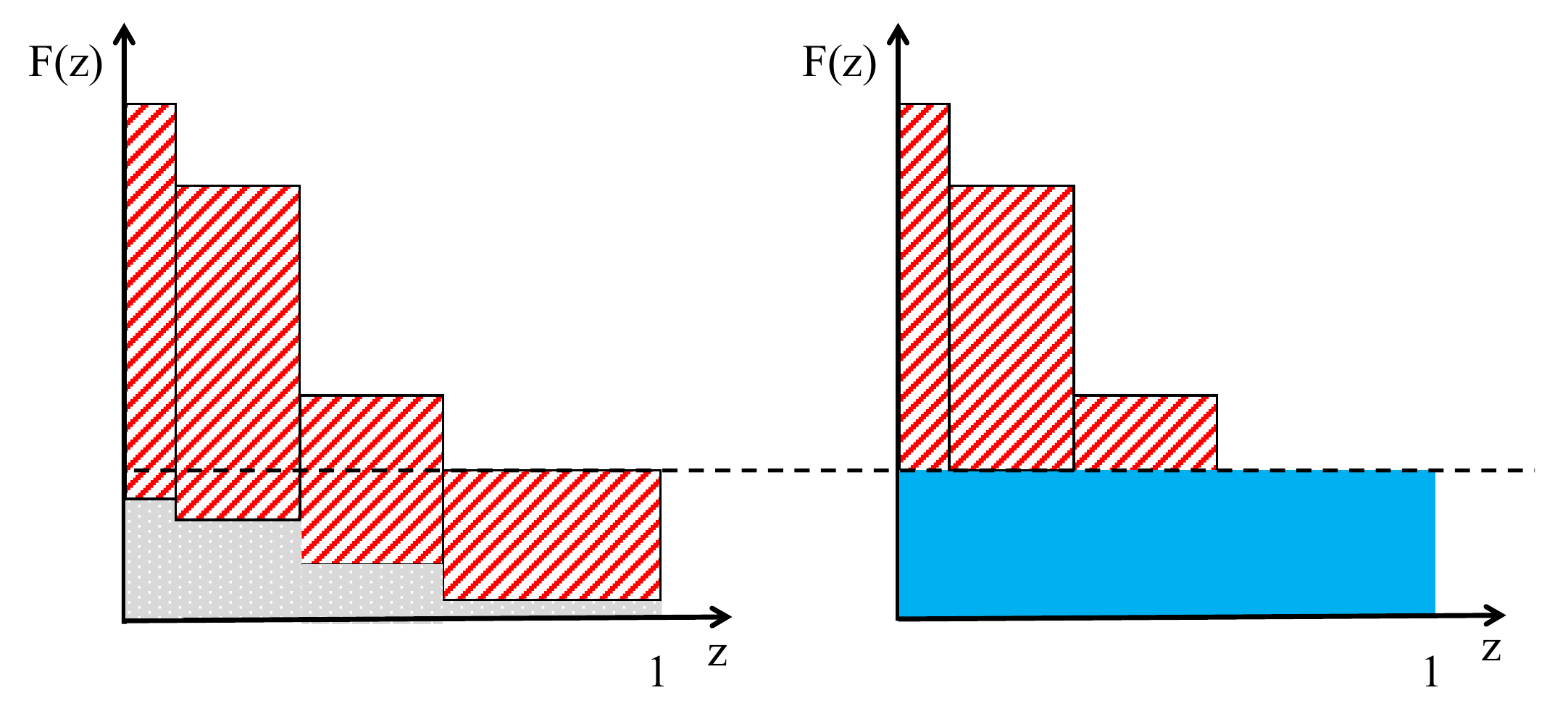}
\caption{Illustration of the first step. In the left the red area illustrates $f(z)$ for $z \in [0, 1]$ and the grey area illustrates $F_r(z)$.  In this step we replace the area below the dotted line with infinitesimal jobs, which are colored blue in the right figure.}
\label{fig:step1}
\end{figure}

Here, it is important to notice that we only replace each job (or rectangle) with smaller jobs (or shorter rectangles). In particular, rectangles corresponding to $f(z)$, $z \in [0, 1]$ can only be replaced into shorter rectangles. This is possible due to the following lemma. 

\begin{lemma}
	For all $z \in [0, 1]$ we have $F_r(z) \leq F(1)$. 
\end{lemma}
\begin{proof}
    By definition of $f$ we have 
    $f(z + 1) \leq f(1)$, $f(z+2) \leq f(2) \ldots$. The claim follows by adding up the inequalities.
\end{proof}

Since $F' = F$, $A'_i$ remains fixed. Further, $O_i$ can only decrease as we replaced each rectangle with shorter rectangles and thus more flexible packings are allowed. Therefore, the first step doesn't decrease $A'_i / O_i$.

\paragraph{Second step.} Since all rectangles corresponding to $f(z)$,  $z \geq 1$ have infinitesimal height, the optimal packing just  `waterfills' the slope defined by $f(z)$, $z \in [0, 1]$. This is shown in the left figure in Figure~\ref{fig:step2}. We would like to increase $A'_i$ keeping $O_i$ the same. Thus, we only modify the slope under the water. Precisely, we find $m' \leq m \in [0, 1]$ such that 
i) $f'(z) = f(z)$ for all $z \in [0, m']$, ii) $f'(z)$ is constant for all $z \in [m', m]$, iii) $f'(z) = 0$ for all $z \in [m, 1]$, and iv) $f'(m) (1 - m)  = \int_{z = 0}^1 F_r(z) dz$. This process is illustrated in Figure~\ref{fig:step2}. Note that this step doesn't change the value of $O_i$.

\begin{figure}[H]
\begin{subfigure}{.50\textwidth}
  \centering
  \includegraphics[width=.7\linewidth]{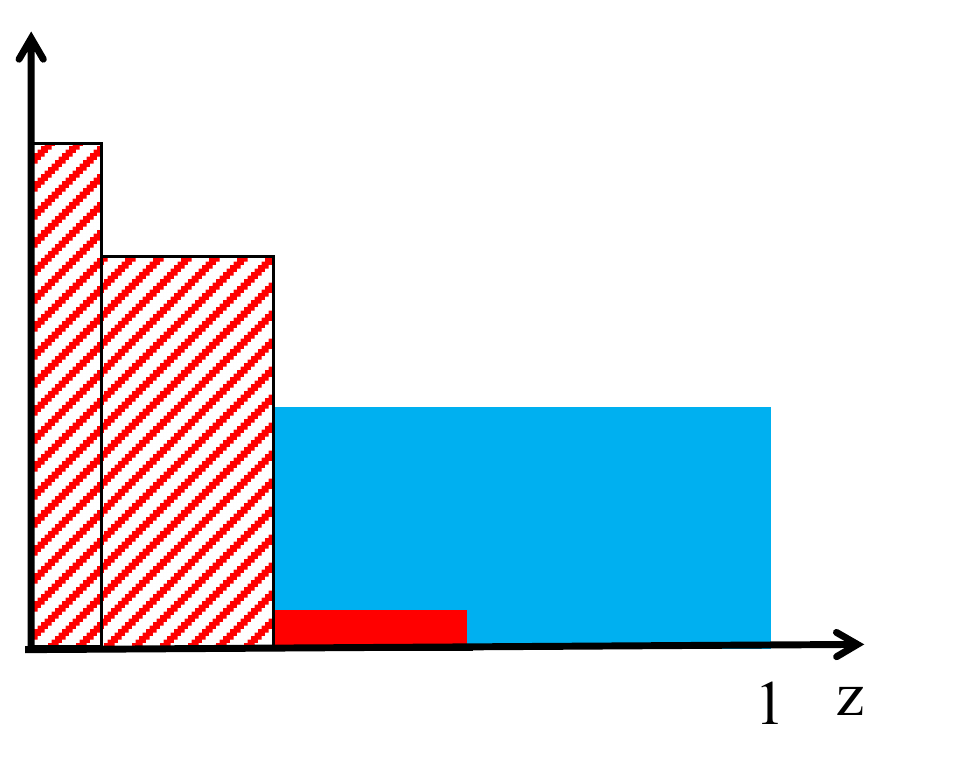}
  \label{fig:sub1}
\end{subfigure}%
\begin{subfigure}{.50\textwidth}
  \centering
  \includegraphics[width=.7\linewidth]{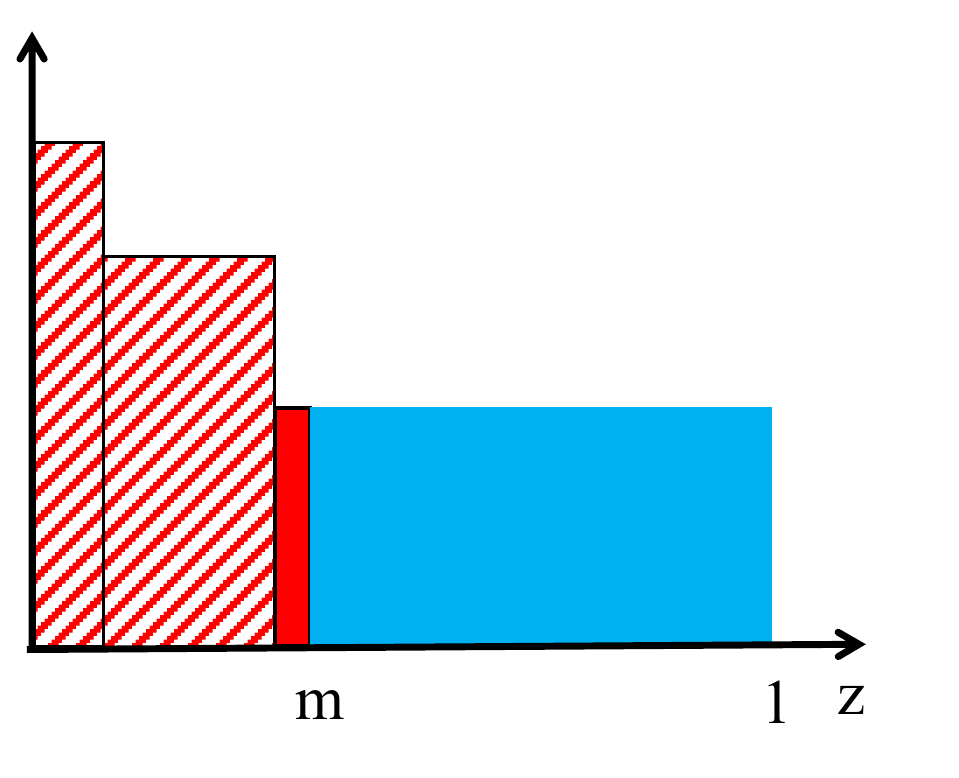}
  \label{fig:step22}
\end{subfigure}
\caption{Illustration of the second step. The left figure shows the optimum packing. The right figure shows the modification: without changing the volume and skyline of infinitesimal rectangles colored blue, shift the volume of $f(z), z \in [0, 1]$ to the left as much as possible. The right figure shows the resulting $F'$, along with $f'(z)$, $z \in [0, 1]$ and $F'_r$.}
\label{fig:step2}
\end{figure}

Further, it is straightforward to see that this step can only increase $A'_i$ as illustrated in Figure~\ref{fig:step2-before-after}.

\begin{figure}[H]
\begin{subfigure}{.50\textwidth}
  \centering
  \includegraphics[width=.7\linewidth]{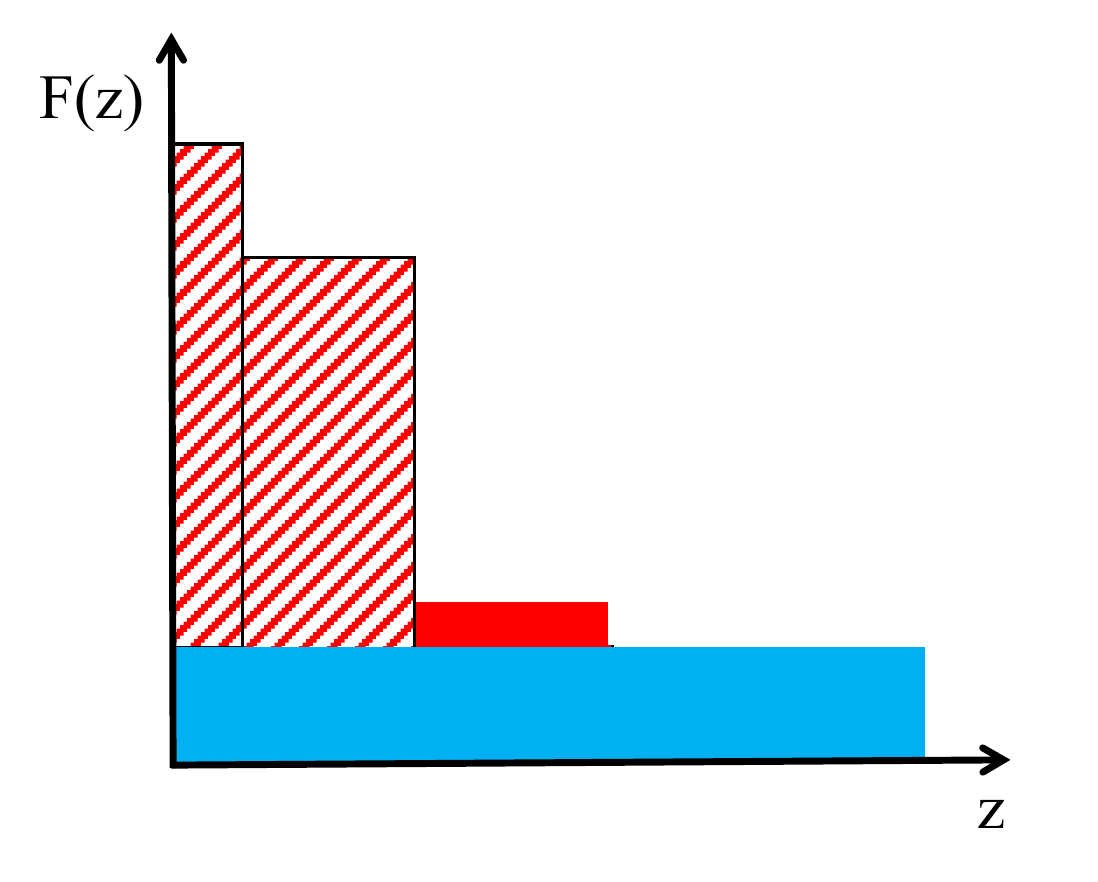}
  \end{subfigure}%
\begin{subfigure}{.50\textwidth}
  \centering
  \includegraphics[width=.7\linewidth]{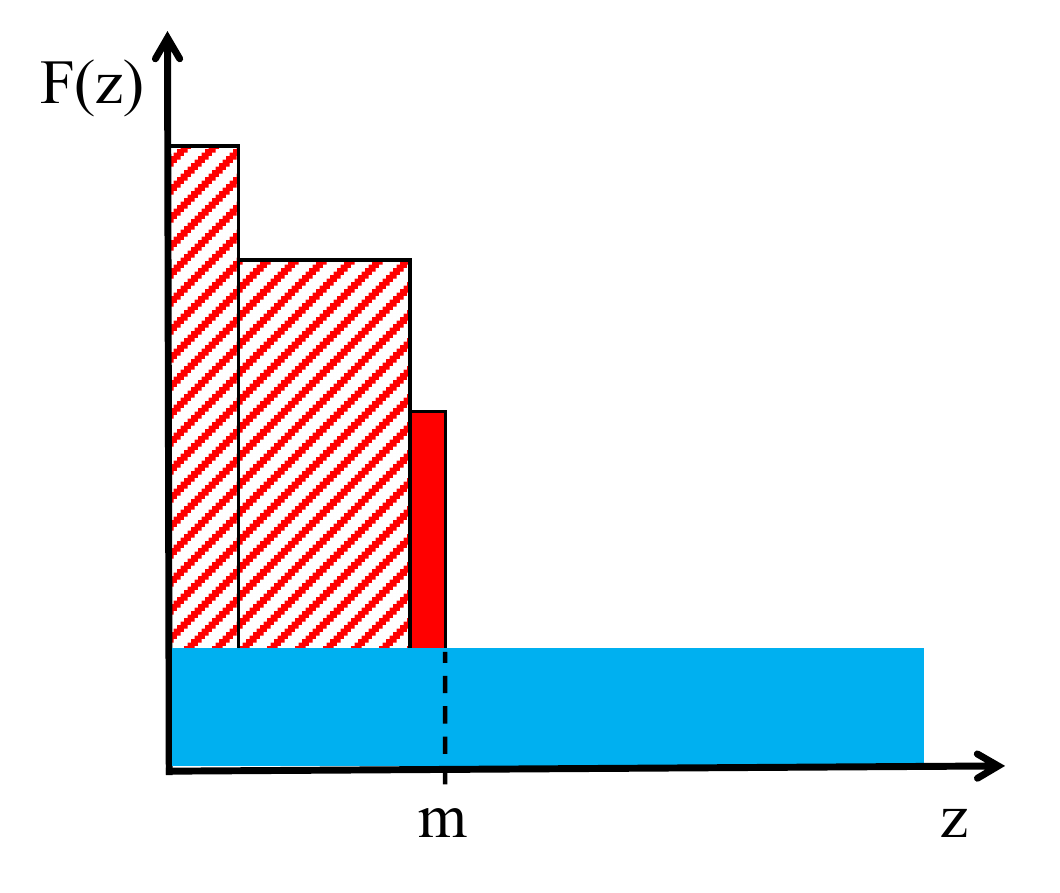}
\end{subfigure}
\caption{Illustration of $F$ before and after the second step.}
\label{fig:step2-before-after}
\end{figure}

\paragraph{Third step.} Change $f$ to $f'$ so that $f'(z)$ is constant for $z \in [0, m]$. In this step we do not change $f(z), z > m$. Here we ensure $A'_i / O_i$ never decreases. Let's see why this is possible. Let $r := F(1)$; in other words, the total height of infinitesimal rectangles in any group of $A_i'$ is $r$. Thanks to the simple structure of $f$, we have 
$$
\frac{A'_i}{O_i} = \frac{\int_{z = 0}^m (r + f(z))^k   dz  + (1 - m)r^k}{\int_{z = 0}^m (f(z))^k   dz	+ (1 - m) (f(m))^k}
$$

Let's consider any two $z_1, z_2 \in [0, m]$ such that $f(z_1) \neq f(z_2)$. Let $I_1$ be the maximal interval including $z_1$ where $f(z)$ is constant for $z \in I_1$. Similarly $I_2$ is defined. We observe that we can make $f$ more uniform without decreasing the value of $A'_i / O_i$.

\begin{lemma}
    It is possible to not to decrease the value of $A'_i / O_i$ by setting $f'(z)  = f(z_1)$ or $f'(z)  = f(z_2)$ for all $z \in I_1 \cup I_2$ ; and $f'(z) = f(z)$ for all $z \notin I_1 \cup I_2$.
\end{lemma}
\begin{proof}
    Let $y_1 = f(z_1)$ and $y_2 = f(z_2)$. Let $\lambda_1 = |I_1|$ and $\lambda_2 = |I_2|$. Then, for some $c_n, c_d > 0$, before the change, we have
    $$
    \frac{A'_i}{O_i} = \frac{ \lambda_1 (y_1 + r)^k + \lambda_2 (y_2 + r)^k + c_n}{\lambda_1 (y_1)^k + \lambda_2 (y_2)^k + c_d }
    $$
    Let $\lambda_1 = \lambda$ and $\lambda_2 = \lambda - \lambda_1$. Then, both the numerator and denominator are linear in $\lambda$. It is an easy exercise to show that $A'_i / O_i$ is maximized when $\lambda = 0$ or $\lambda = \lambda_1 + \lambda_2$.
\end{proof}

The above process of making the function values uniform must terminate in a finite time steps assuming that $f$ has a finite number of distinct function values. Thus, at the end $f(z)$ is constant for $z \in [0, m]$ and $f(z) = 0$ for $z \in [m, 1]$. Technically speaking $f(z)$ should be a sufficiently small constant for $z \in [m, 1]$ to capture infinitesimal sizes, but we pretend that the value is zero for convenience.

After all the modification steps, the final $f$ is illustrated in Figure~\ref{fig:step23-4}.

\begin{figure}[H]
\begin{subfigure}{.50\textwidth}
  \centering
  \includegraphics[width=.7\linewidth]{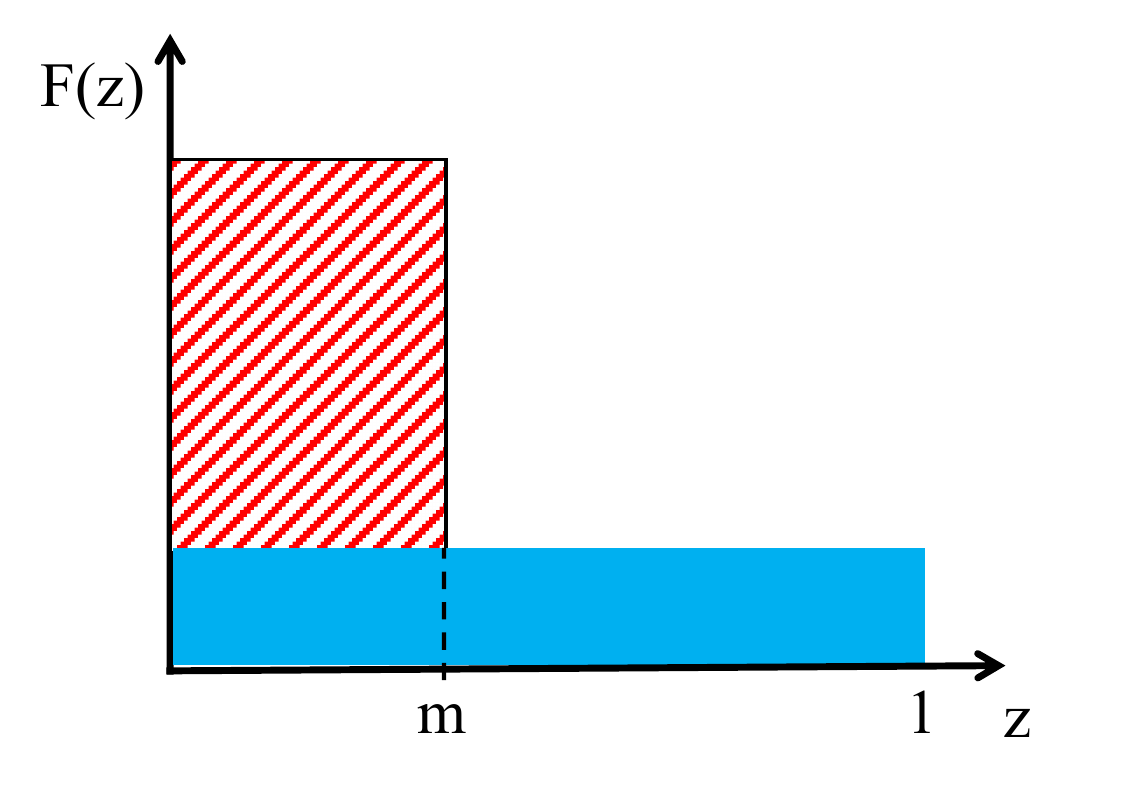}
  \end{subfigure}%
\begin{subfigure}{.50\textwidth}
  \centering
  \includegraphics[width=.6\linewidth]{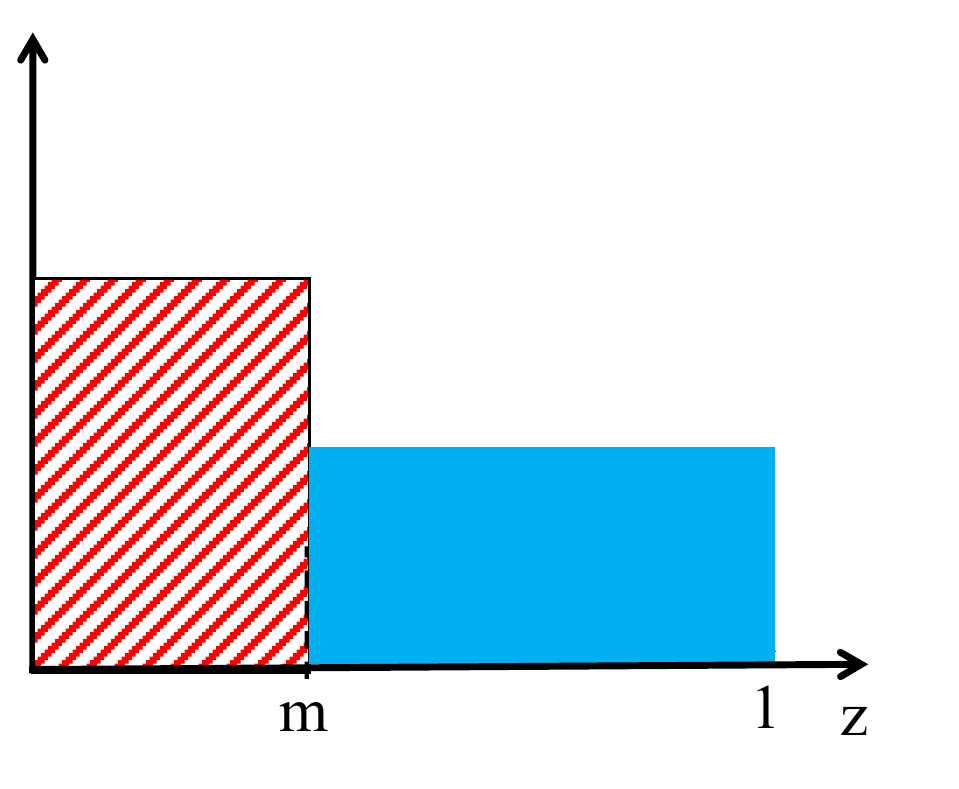}
\end{subfigure}
\caption{After all the modification steps, $f(z), F_r(z), z\in [0, 1]$ are colored in red and blue, respectively. The right figure shows the optimal packing.}
\label{fig:step23-4}
\end{figure}

To upper bound $A'_i / O_i$, by scaling we can assume wlog that 
\begin{align*}
f(z) &= 
\begin{cases}
    1 & z \in [0, m] \\
    0 & z \in [m, 1], 
\end{cases}\\
F_r(z) &= (1- m) y \quad z \in [0, 1]
\end{align*}
for some $y \in [0, 1]$. Define 

\begin{equation}
    \label{eqn:g}
g(m, y) := \frac{ m \cdot (1+ (1 - m)y)^k + (1- m) \cdot ((1- m)y)^k}{m \cdot 1^k + (1 - m) \cdot y^k}
\end{equation}

Note that $A'_i / O_i = g(m, y)$. Thus, thanks to Lemma~\ref{lem:lk-middle-lemma} we have Theorem~\ref{thm:Lk}, which we reproduce below for easy reference. 

\begin{theorem}
   There is a randomized algorithm that is 
    $\alpha^{1/k}$-approximate for minimizing the $L_k$-norm of machine loads, where 
    $\alpha := \sup_{m \in (0, 1], y \in [0, 1]} g(m, y)$.
\end{theorem}

As mentioned before, we can compute an upper bound on the claimed approximation ratio $\alpha^{1/ k}$ for small values of $k$. See Table~\ref{table:k}. When $k = 2$, we analytically show that $\alpha = 4/3$ in Appendix~\ref{sec:k-2}.

\bibliographystyle{plain}
\bibliography{unrelated}

\appendix
\section{Missing Proofs in Section~\ref{sec:SNC}}
\label{appendix:proofs-SNC}
\breakintosegments*
\begin{proof}
	If the total length of $C$ is an integer multiply of $4$, then we can partition the edges of $C$ into segments of length $4$. There are two different ways to do this, and we output one uniformly at random. Clearly, every group in $C$ is the center of a segment with probability $1/2$. 
	
	Otherwise, the length of $C$ is $4z + 2$ for some integer $z \geq 1$. If the length of $C$ is at least $18$, we do the following. We randomly choose a job in $j \in C \cap J$ from $C$, remove it and its two adjacent edges from $C$, break the resulting path of length $4z$ into $z$ segments of length $4$, and output the $z$ segments.  The probability that a vertex is the center of a segment is $\frac{z}{2z + 1} \geq \frac49$.
	
	It remains to consider the case where the length of $C$ is $10$ or $14$. In this case, we show that $C$ contains a valid segment $S$ of length $6$. As two groups $u$ and $u'$ adjacent to a same job $j$ in $C$ have $g(u) \neq g(u')$, the multi-set $M':=\{g(u): u \in C \cap U\}$ contains at least $3$ distinct machines. 
	\begin{itemize}
		\item If $C$ has length $10$, then $M'$ has size $5$ and every machine appears at most twice in $M'$.  Among all the 5 sub-paths of $C$ of length $6$ that start and end at groups, at most $4$ of them are invalid segment. This holds since a pair of the same machine in $M'$ can invalidate two segments of length $6$.
		\item If $C$ has length $14$, then $M'$ has size $7$ and a machine appears at most 3 times in $M'$. If a machine appears 3 times, then the three groups $u, u', u''$ with $g(u) = g(u') = g(u'')$ breaks $C$ into paths of length $4$, $4$ and $6$ respectively. The path of length $6$ is a valid segment as the two end-groups have the same $g$ value.  Now assume no machine appears 3 times in $M'$, then as before we can invalidate at most $6$ segments of length $6$, and there is one (valid) segment of length $6$ in $C$. 
	\end{itemize}
	With the segment $S$ of length $6$, we show how to output the segments.  With probability $1/2$, we output $S$, breaking the remaining edges into segments of length $4$ (there are 1 or 2 of them) and output these segments also. With the remaining probability $1/2$,  we do the following. Let $j$ be the job between the two centers of $S$. We remove $j$ and its two adjacent edges from $C$, and break the remaining path into segments of length $4$.   Clearly one can show that every group is a center of some segment with probability $1/2$. 
\end{proof} 

		\functionconvex*
		\begin{proof}
			The derivative of the function is $\frac{e^{-x}}{1-e^{-x}} - \frac1x = \frac1{1-e^{-x}} - 1 - \frac1x$. The second-order derivative is $\frac{-e^{-x}}{(1-e^{-x})^2} + \frac1{x^2}$. It remains to show that $(1-e^{-x})^2 \geq e^{-x}x^2$ for every $x \geq 0$, which is equivalent to $e^{-x}  + e^{-x/2}x \leq 1$ for every $x \geq 0$.  The inequality holds with equality when $x = 0$. The derivative of the function is $-e^{-x} - e^{-x/2}x/2 + e^{-x/2} = e^{-x/2}(1 - e^{-x/2} - x/2) \leq 0$. 
		\end{proof}

\section{Missing Proofs in Section~\ref{sec:wC-algorithm}}
	\applyingSN*
	\begin{proof}
		$U$, $g$ and $y$ are determined when we fixed $\rho$ and $\vec{\tau}$. To see \ref{property:applying-SN-marginal}, notice that we assign $j$ to a group $u$ with probability $y_{uj}$, which is the total height of rectangles in $u$ for $j$. Under this event, we choose the anchor rectangle from them at random, proportional to their heights.
		
		Now we consider \ref{property:applying-SN-independence}. Let $u$ be the group containing $R_{ijs}$; define $u'$ similarly for $R_{ij's'}$. 
		We have $g(u) = g(u') = i$. \ref{property:SN-independence} in Theorem~\ref{thm:SN} says that $\Pr[\sigma(j) = u, \sigma(j') = u' | \rho, \vec{\tau}] \leq y_{uj} y_{u'j'}$.  Then $\Pr[R_{ijs}, R_{ij's'}|\rho, \vec \tau, \sigma(j) = u, \sigma(j') = u'] = \frac{x_{ijs}}{y_{uj}}\cdot \frac{x_{ij's'}}{y_{u'j'}}$ by the way we choose anchor rectangles.  Combining the two inequalities gives $\Pr[R_{ijs}, R_{ij's'} | \rho, \vec \tau] \leq x_{ijs}x_{ij's'}$ for any $\rho$ and $\vec \tau$. Then de-conditioning on $\rho$ and $\vec \tau$ proves \ref{property:applying-SN-independence}.
		
		Finally we turn to \ref{property:applying-SN-SN}. Assume $i$ does not dominate any of $j$ and $j'$. Fix a base window $k$. For a fixed $\rho$ and $\vec{\tau}$ for which $R_{ijs} \sim_k R_{ij's'}$, we have $\Pr[\sigma(j)  = u_{ik}, \sigma(j')  = u_{ik} | \rho, \vec{\tau}] \leq (1 - \eta) y_{u_{ik}j} y_{u_{ik}j'}$, by the Property \ref{property:SN-SN} in Theorem~\ref{thm:SN}. Again by the way we choose the anchor rectangles, we have 
		\begin{align*}
			\Pr[R_{ijs}, R_{ij's'} | \rho, \vec{\tau}] &\leq \Pr[\sigma(j)  = u_{ik}, \sigma(j')  = u_{ik} | \rho, \vec{\tau}]\cdot \Pr[R_{ijs}, R_{ij's'} | \sigma(j)  = u_{ik}, \sigma(j')  = u_{ik}, \rho, \vec{\tau}] \\
			&\leq (1 - \eta) y_{u_{ik}j} y_{u_{ik}j'} \cdot \frac{x_{ijs}}{y_{u_{ik}j}}\cdot \frac{x_{ij's'}}{y_{u_{ik}j'}} = (1-\eta) x_{ijs} x_{ij's'}.
		\end{align*}
		Then de-conditioning on $\rho$ and $\vec{\tau}$, we have $\Pr[R_{ijs}, R_{ij's'} | R_{ijs}\sim_k R_{ij's'}] \leq (1-\eta) x_{ijs} x_{ij's'}$. De-conditioning on $k$ proves \ref{property:applying-SN-SN}.\hfill
	\end{proof}	

\section{Missing Proofs in Section~\ref{sec:wC-analysis}}
	\casedominates*
	
	\begin{proof} For every $j \in J$ and real $t \geq 0$, let $\vol_{i, j}(t)$ be the total volume of portions of rectangles on $i$ for $j$ before time $t$. For every $i, J' \subseteq J$ and $t$, let $\vol_{i,J'}(t) = \sum_{j \in J'} \vol_{i, j}(t)$. Then, 
		\begin{align*}
			\E[C^* | i_{j^*}  = i] - p_{ij^*} &\leq \frac{1}{x_{ij^*}}\sum_{s^*} x_{ij^*s^*}\int_{0}^{p_{ij^*}} \vol_{i, J \setminus j^*}((1+\alpha)s^*+\tau^*+0.2 \cdot  p_{ij^*})\frac{\sfd\tau^*}{p_{ij^*}}\\
			&\leq \frac{1}{x_{ij^*}}\sum_{s^*} x_{ij^*s^*}\int_{0}^{p_{ij^*}} \big((1+\alpha)s^*+\tau^* + 0.2 \cdot p_{ij^*} - \vol_{i,j^*}(s^*+\tau^*)\big)\frac{\sfd\tau^*}{p_{ij^*}}\\
			&= \frac{1}{x_{ij^*}}\sum_{s^*} x_{ij^*s^*}\left((1+\alpha)s^*+\frac{p_{ij^*}}{2}  + 0.2 \cdot  p_{ij^*}- \int_{0}^{p_{ij^*}}  \vol_{i,j^*}(s^*+\tau^*)\frac{\sfd\tau^*}{p_{ij^*}}\right)\\
			&= \frac{1+\alpha}{x_{ij^*}}\sum_{s^*} x_{ij^*s^*}\cdot s^* +0.7p_{ij^*} - \frac{x_{ij^*}p_{ij^*}}{2}\\
			&\leq  \frac{1+\alpha}{x_{ij^*}}\sum_{s^*} x_{ij^*s^*}\cdot s^* +0.45p_{ij^*}.
		\end{align*}
		We explain the inequality in the first line. The probability of $R_{ij^*s^*}$ is $\frac{x_{ij^*s^*}}{x_{ij^*}}$, conditioned on $i_{j^*}  = i$.  As $i$ dominates $j^*$, we have $\theta_{j^*} = (1+\alpha) s^* + \tau^* + 0.2 \cdot  p_{ij^*}$ conditioned on $R_{ij^*s^*}$ and $\tau_{ij^*} = \tau^*$. Under these conditions, the expected starting time of $j^*$ is at most $\vol(\theta_{j^*})$: for a rectangle $R_{ijs}$ with $j \neq j^*$, we have $\Pr[R_{ijs}, \theta_{ij} < \theta_{j^*} | R_{ij^*s^*}, \tau_{ij^*} = \tau^*] \leq \Pr[R_{ijs},s + \tau_{ij} < \theta_{j^*} | R_{ij^*s^*}, \tau_{ij^*} = \tau^*] \leq x_{ijs}\cdot \Pr[s + \tau_{ij} < \theta_{j^*} | R_{ij^*s^*}, \tau_{ij^*} = \tau^*]$. This is $\frac{1}{p_{ij}}$ times the volume of the rectangle $R_{ijs}$ before time $\theta_j^*$. If the event happens, the rectangle contributes $p_{ij}$ to the starting time of $j^*$.  So the inequality holds.
		
		The second inequality used that the total volume of all rectangles on machine $i$ before time $(1+\alpha)s^* + \tau^* + \alpha p_{ij^*}$ is at most $(1+\alpha)s^* + \tau^* + 0.2 \cdot  p_{ij^*}$, and $\vol_{i, j^*}(s^* + \tau^*) \leq \vol_{i, j^*}((1+\alpha)s^* + \tau^* + 0.2 \cdot  p_{ij^*})$.
		
		We then prove 
			\begin{align*}
				\frac{1}{x_{ij^*}}\sum_{s^*}x_{ij^*s^*} \frac{1}{p_{ij^*}}\int_0^{p_{ij^*}}\vol(s^*+\tau^*)\sfd\tau^* = \frac{x_{ij^*}p_{ij^*}}{2},
			\end{align*}
		which is used in the equality in the fourth line. Notice $\sum_{s^*, s'^*}\int_0^{p_{ij^*}}\int_0^{p_{ij^*}} x_{ij^*s^*}x_{ij^*s'^*}{\bf1}(s'^* + \tau'^* < s^* + \tau^*) \sfd\tau^*\sfd \tau'^* = \frac{x_{ij^*}^2p_{ij^*}^2}{2}$. Scaling both sides by $\frac1{x_{ij^*}p_{ij^*}}$ gives the inequality. 
		
		The last inequality used that $x_{ij^*} > \frac12$ since $i$ dominates $j^*$.  So, $\E[\tilde C_{j^*} | i_{j^*}  = i] \leq \frac{1}{x_{ij^*}}\sum_{s^*} x_{ij^*s^*}\cdot \left((1+\alpha)s^* +1.45p_{ij^*}\right) \leq \frac{1.45}{x_{ij^*}}\sum_{s^*}x_{ij^*s^*}(s^* + p_{ij^*})$. \hfill
	\end{proof}

	\ECbyPhif*
	\vspace*{-20pt}
    \begin{align*}
	    \textit{Proof.} && \widehat\E[C^*] - p_{ij^*} &= \sum_{j \neq j^*, s} \widehat\Pr[R_{ijs}, \theta_j < \theta^*] \cdot p_{ij} &&\\
		&& &= \sum_{j \neq j^*, s} \widehat\Pr[R_{ijs}, \theta_j < \theta^*]\cdot p_{ij} \cdot \frac{1}{x_{ijs}}\cdot  \sum_{f \in \calF: R_{ijs} \in f} z_f &&\\
		&& &= \sum_{f \in \calF} \sum_{j \neq j^*, s:R_{ijs} \in f} z_f \cdot \frac{1}{x_{ijs}} \cdot \widehat\Pr[R_{ijs}, \theta_j < \theta^*] \cdot p_{ij} = \sum_{f \in F} \Phi(f). && 
    \end{align*}

\Phiftozf*
\begin{proof}
	We first give a coarse upper bound for the first term in the min operator, so that we can use it to absorb some bound we obtain along the way. 
	\begin{align*}
		&\quad \eta\cdot \E_\rho \int_0^{\theta^*} {\bf1}\left(s^*< \frac{h_\rho(s^* + \tau^*)}{1+\beta} <  \tau < h_\rho(s^* + \tau^*)\right)\cdot \sfd\tau \\
		&\leq \eta\cdot \E_\rho \int_0^{\theta^*} {\bf1}\left(\frac{h_\rho(\theta^*/(1+\alpha))}{1+\beta} < \tau\right)\cdot \sfd\tau  = \frac{\eta}{\ln(1+\beta)}\int_{1}^{1+\beta}\frac{\sfd\varrho}{\varrho}\cdot \left(\theta^* - \frac{\varrho\theta^*}{(1+\alpha)(1+\beta)}\right)\\
		&=\frac{\eta}{\ln(1+\beta)}\left(\theta^*\ln(1+\beta) - \frac{\beta\theta^*}{(1+\alpha)(1+\beta)}\right) = \eta\theta^*\left(1 - \frac{\beta}{(1+\alpha)(1+\beta)\ln(1+\beta)}\right) \leq 0.8\eta \theta^*.
	\end{align*}
	The inequality comes from that the event in $\bf1(\cdot)$ on the right side of the inequality is easier to satisfy than that on the left side. The first equality follows by defining $\varrho = \frac{h_\rho(\theta^*/(1+\alpha))}{\theta^*/(1+\alpha)}$ and that $\ln \varrho$ is uniformly distributed in $[0, \ln(1+\beta)]$.
	
	For a rectangle $R_{ijs} \in f$ with $s + p_{ij} \leq \frac{\theta^*}{1 + \alpha}$, we have $\widehat\Pr[R_{ijs}, \theta_j < \theta^*] \leq \widehat\Pr[R_{ijs}] \leq x_{ijs}$ due to \ref{property:applying-SN-independence}. Thus $\frac{1}{x_{ijs}}\widehat\Pr[R_{ijs}, \theta_j < \theta^*] \leq 1$. For a rectangle $R_{ijs}$ with $s > \frac{\theta^*}{1 + \alpha}$, we have $\widehat\Pr[R_{ijs}, \theta_j < \theta^*] = 0$ as $(1+\alpha)s + \tau_{ij} \geq (1+\alpha)s > \theta^*$ with probability $1$.  Let $R_{ijs}$ be the unique rectangle with $s \leq \frac{\theta^*}{1+\alpha} < s + p_{ij}$ in the configuration. We assume $R_{ijs}$ exists: Otherwise $\frac{\Phi(f)}{z_f} \leq \frac{\theta^*}{1 + \alpha} \leq  \theta^* - \eta\theta^*$, and $\eta \theta^*$ is at least the first term in the min operator in \eqref{inequ:bound-Phi-f}.  So \eqref{inequ:bound-Phi-f} holds.

	For notational convenience, we use $p$ for $p_{ij}$ and $\tau$ for $\tau_{ij}$.  The total contribution of rectangles in $f$ before $R_{ijs}$ to $\Phi(f)$ is at most $z_f s$, and the total contribution rectangles in $f$ after $R_{ijs}$ is $0$.  So, 
	\begin{align*}
		\frac{\Phi(f)}{z_f} &\leq s + \frac{1}{x_{ijs}}\cdot \widehat\Pr[R_{ijs}, \theta_j < \theta^*]\cdot p.
	\end{align*}
	If $i$ dominates $j$, then $\frac{\Phi(f)}{z_f} \leq s + (\theta^* -(1+\alpha)s - 0.2 p)_+ = \max\{s, \theta^* - \alpha s - 0.2\cdot p\}$, as $\theta_j = (1+\alpha)s + \tau + 0.2 p$ and $\tau$ is uniformly distributed in $[0, p)$.  Notice $s \leq \frac{\theta^*}{1+\alpha} \leq \theta^* - \eta \theta^*$ and $\alpha s + 0.2 \cdot p \geq 0.8 \eta\cdot (1+\alpha)(s + p) \geq 0.8\eta\theta^*$. Thus $\frac{\Phi(f)}{z_f} \leq \theta^* - 0.8\eta\theta^*$. Again $0.8\eta\theta^*$ is lower bounded by the first term  in the min operator of \eqref{inequ:bound-Phi-f}. 
	
	From now on we assume $i$ does not dominate $j$.  Let $q = \min\{\theta^* - (1+\alpha)s, p\} \geq 0$. 
	For any $\tau \in [0, p]$, we have  $(1+\alpha)s + \tau < \theta^*$ if and only if $\tau < q$. 
	For any $\tau \in [0, q)$, we have 
	\begin{align*}
		&\quad \widehat\Pr[R_{ijs}, \theta_j  < \theta^* | \tau] = \widehat\Pr[R_{ijs} | \tau] \\		 	
		&= \widehat\Pr[R_{ijs} \sim R_{ij^*s^*}|\tau]\cdot \widehat\Pr[R_{ijs}|R_{ijs} \sim R_{ij^*s^*}, \tau] + \widehat\Pr[R_{ijs} \not\sim R_{ij^*s^*}|\tau]\cdot \widehat\Pr[R_{ijs}|R_{ijs} \not\sim R_{ij^*s^*}, \tau]\\
		& \leq  \widehat\Pr[R_{ijs} \sim R_{ij^*s^*}|\tau]\cdot (1-\eta) x_{ijs} + \widehat\Pr[R_{ijs} \not\sim R_{ij^*s^*}|\tau]\cdot x_{ijs}\\
		&= x_{ijs} \big( 1 - \eta \cdot \widehat \Pr[R_{ijs} \sim R_{ij^*s^*} | \tau]\big).
	\end{align*}
	So, we have
	\begin{align}
		\frac{\Phi(f)}{z_f}&= s + \frac 1p\int_{0}^q (1 - \eta\cdot\widehat\Pr[R_{ijs} \sim R_{ij^*s^*}|\tau])\cdot p\cdot\sfd\tau \quad=\quad s + q - \eta \int_0^q \widehat\Pr[R_{ijs} \sim R_{ij^*s^*}|\tau]\cdot \sfd\tau \nonumber\\
		&= s + q - \eta \int_0^q \widehat\Pr[R_{ijs} \sim R_{ij^*s^*}|\tau]\sfd\tau \quad = \quad  s + q - \eta \cdot \E_\rho \int_0^q \widehat\Pr[R_{ijs} \sim R_{ij^*s^*}|\tau,\rho]\cdot\sfd\tau \nonumber\\
		&= s + q - \eta\cdot \E_\rho \int_0^q {\bf1}\left(\max\{s, s^*\} < \frac{h_\rho(s^* + \tau^*)}{1+\beta} < s + \tau < h_\rho(s^* + \tau^*)\right)\cdot \sfd\tau. \label{equ:quantity-for-Phi/z}
	\end{align}
	
	In the rest of the proof, we fix $s^* \geq 0, \tau^* \geq 0$ and $\theta^* = (1+\alpha)s^* + \tau^*$, and upper bound \eqref{equ:quantity-for-Phi/z} subject to $s, q \geq 0, (1+\alpha)s + q \leq \theta^*$.  It is easy to see that if we increase $q$ continuously at a rate of $1$, the negative term increases by at most a rate of $\eta$.  So the quantity increases. Thus, we can assume $(1+\alpha)s + q = \theta^*$.	
	
	
	We consider two cases. In the first case we assume $s > s^*$, or $(1+\beta)s\leq (s^* + \tau^*)$, which implies that the condition $s < \frac{h_\rho(s^* + \tau^*)}{1+\beta}$ inside $\bf1(\cdot)$ is redundant. We apply the following continuous operation:  decrease $s$ at a rate of $1$ and increase $q$ at a rate of $1+\alpha$ so that $(1+\alpha)s + q = \theta^*$ holds in the process. This increases $s+q$ at a rate of $\alpha$. For a fixed $\rho \in [1, 1+\beta)$, we consider the interval of $\tau$ values in $[0, q]$ for which the condition inside $\bf1(\cdot)$ holds. The length of the interval increases at a rate of at most $1+\alpha$, which is how fast $q$ increases. So the how negative term increase at  a rate of at most $(1+\alpha)\eta$.  As $\alpha > (1+\alpha)\eta$, the operation will increase \eqref{equ:quantity-for-Phi/z}.  
	
	So, the quantity is maximized when $s = 0$ when $s \leq s^*$. In this case, we have $q = \theta^*$ and 
	\begin{flalign*}
		&& \frac{\Phi(f)}{z_f} \leq \theta^* - \eta\cdot \E_\rho \int_0^{\theta^*} {\bf1}\left(s^*< \frac{h_\rho(s^* + \tau^*)}{1+\beta} < \tau < h_\rho(s^* + \tau^*)\right)\cdot \sfd\tau. &&
	\end{flalign*}
	This is what we obtain using the first term in the min operator of \eqref{inequ:bound-Phi-f}.
	
	Then we consider the second case where we have $s \geq s^*$ and $(1+\beta)s > (s^* + \tau^*)$. As $\theta^* = (1+\alpha)s^* + \tau^* = (1+\alpha)s + q$, we have $s + q \leq s^* + \tau^* < (1+\beta)s$.  We have
	\begin{align*}
		&\quad s+q  - \eta\cdot \E_\rho \int_0^q {\bf1}\left(\max\{s, s^*\} < \frac{h_\rho(s^* + \tau^*)}{1+\beta} < s + \tau < h_\rho(s^* + \tau^*)\right)\cdot \sfd\tau\\
		&= s + q - \eta\cdot \E_\rho \int_0^q {\bf1}\left(s < \frac{h_\rho(s^* + \tau^*)}{1+\beta} < s + \tau\right)\cdot \sfd\tau \\
		&=s + q - \frac{\eta }{\ln(1+\beta)}\int_{1}^{\frac{s+q}{s}} \left(s + q - \varrho s\right)\cdot \frac{\sfd \varrho}{\varrho}\\
		&=s + q - \frac{\eta}{\ln(1+\beta)}\left((s+q)\ln\frac{s+q}{s} - q\right) \\
		&=\theta^* - \left(\alpha s + \frac{\eta}{\ln(1+\beta)}\left((\theta^* - \alpha s)\ln\frac{\theta^* - \alpha s}{s} - \theta^* + (1+\alpha)s\right)\right).
	\end{align*}
	The second inequality follows by defining $\varrho$ so that $\frac{h_\rho(s^* + \tau^*)}{1+\beta} = \varrho s$. Notice that $\ln \varrho$ is uniformly distributed in $\left(\ln\frac{s^* + \tau^*}{(1+\beta)s}, \ln\frac{s^* + \tau^*}{s}\right]$, and the integration is $s + q - \varrho s$ if $s < \varrho s < s + q$ and 0 otherwise.  The last equality used that $\theta^* = (1+\alpha) s + q$.

	Let $x = \frac{s}{\theta^*} \in [0, 1]$. The minimum of $\alpha x + \frac{\eta}{\ln(1+\beta)}\left((1 - \alpha x)\ln\frac{1 - \alpha x}{x} - 1 + (1+\alpha)x\right)$ over $x \in [0, 1]$ is at least $0.102$. See Figure~\ref{fig:plotting-3} in Appendix~\ref{appendix:wC-plotting} for the plotting of the function. Therefore, we have $\frac{\Phi(f)}{z_f} \geq \theta^* - 0.102\theta^*$. This corresponds to the second term in the min operator of \eqref{inequ:bound-Phi-f}. \hfill
\end{proof}

\Qcircr*
\begin{proof}
		For convenience, we recall that $Q^\circ(r) = \eta\cdot\E_\rho \int_{0}^{1+\alpha + r} {\bf1}\Big(1 < \frac{h_\rho(1 + r)}{1+\beta} < y < h_\rho(1 + r)\Big)\cdot \sfd y$.
		
				If $r \in [0, \beta - \alpha]$, we have $1+\alpha + r \leq 1 + \alpha + \beta - \alpha = 1 + \beta$. 		The condition $y < h_\rho(1 + r)$ in the definition of $Q^\circ(r)$ is redundant since if $1 < \frac{h_\rho(1 + r)}{1 + \beta}$, then $h_\rho(1 + r) > 1 + \alpha + r$. We define $\varrho := \frac{h_\rho(1 + r)}{1 + \beta}$, and thus $\ln\varrho$ is uniformly distributed in $\big[\ln\frac{1 + r}{1 + \beta},\ln(1 + r)\big)$. Notice that $1\leq 1 + r \leq 1+\beta$. We have $\int_0^{1 + \alpha + r}{\bf1}\big(1 < \frac{h_\rho(1 +r)}{1+\beta} <  y \big)\sfd y $ is $1 + \alpha + r - \varrho$ if $\varrho > 1$ and $0$ otherwise.	
	\begin{flalign*}
		&&Q^\circ(r) = \frac{\eta}{\ln(1+\beta)}\int_{1}^{1 + r} (1 + \alpha + r - \varrho)\frac{\sfd\varrho}{\varrho} = \frac{\eta}{\ln(1+\beta)}\big((1 + \alpha + r)\cdot\ln(1+r) - r\big). && 
	\end{flalign*} \smallskip

	Then we consider the case $r \in (\beta - \alpha, \beta]$. We have $1+r \leq 1+\beta \leq 1 + \alpha + r$.  We define $\varrho := \frac{h_\rho(1 + r)}{1 + \beta}$.   Then, $\int_0^{1 + \alpha + r} {\bf1}\left(1 < \frac{h_\rho(1 + r)}{1+\beta} <  y < h_\rho(1 + r)\right)\cdot \sfd y $ is $\beta\varrho$ if $\varrho \in \big[1, \frac{1 + \alpha + r}{1 + \beta}\big)$, and $1 + \alpha + r  - \varrho$ if $\varrho \in \big[\frac{1 + \alpha + r}{1 + \beta}, 1+r \big]$. (Notice that $\frac{1 + \alpha + r}{1 + r} \leq 1+\alpha \leq 1+\beta$.) So, 
\begin{align*}
	Q^\circ(r) &= \frac{\eta}{\ln(1+\beta)} \left(\int_{1}^{\frac{1 + \alpha + r}{1 + \beta}}\beta\varrho\frac{\sfd\varrho}{\varrho} + \int_{\frac{1 + \alpha + r}{1 + \beta}}^{1+r}(1 + \alpha + r - \varrho)\frac{\sfd\varrho}{\varrho}\right)\\
	&= \frac{\eta}{\ln(1+\beta)}\left(\frac{\beta(1 + \alpha + r)}{1+\beta} - \beta + (1 + \alpha + r)\cdot \ln \frac{(1+\beta)(1+r)}{1 + \alpha + r} - \Big(1 + r - \frac{1 + \alpha + r}{1+\beta}\Big)\right) \\
	&=\frac{\eta}{\ln(1+\beta)}\left(1 + \alpha + r - (1+\beta) - r + (1 + \alpha + r) \ln\frac{(1 + \beta)(1 + r)}{1 + \alpha + r}\right)\\
	&=\frac{\eta}{\ln(1+\beta)}\left((1 + \alpha + r) \ln\frac{(1 + \beta)(1 + r)}{1 + \alpha + r} - (\beta - \alpha)\right).
\end{align*}

Finally consider the case $r > \beta$. 	In this case we have $1 + \alpha + r \geq 1 + r \geq 1 + \beta$. So, $1 < \frac{h_\rho(1 + r)}{1+\beta}$ is always satisfied and the requirement can be removed in the definition of $Q(r)$ in \eqref{equ:define-Q}. We define $\varrho = \frac{h_\rho(1 + r)}{1 + r}$, which takes values in $[1, 1+\beta)$.  Then 
$\int_0^{1 + \alpha + r} {\bf1}\left(\frac{h_\rho(1 + r)}{1+\beta} <  y < h_\rho(1 + r)\right)\cdot \sfd y $ is $\frac{\beta}{1+\beta}\cdot \varrho (1 + r)$ 
if $\varrho \in \big[1, \frac{1 + \alpha + r}{1 + r}\big)$, and $1 + \alpha + r  - \frac{\varrho (1 + r)}{1+\beta}$ if $\varrho \in \big[\frac{1 + \alpha + r}{1 + r}, 1 + \beta\big)$. So, 	 
\begin{flalign*}
	&& Q^\circ(r) &= \frac{\eta}{\ln(1+\beta)}\left(\int_1^{\frac{1 + \alpha + r}{1 + r}}\frac{\beta}{1+\beta}\cdot \varrho(1 + r)\cdot \frac{\sfd\varrho}{\varrho} + \int_{\frac{1 + \alpha + r}{1 + r}}^{1+\beta}\Big(1 + \alpha + r - \frac{\varrho(1 + r)}{1+\beta}\Big)\frac{\sfd\varrho}{\varrho} \right) &&\\
	&& &= \frac{\eta}{\ln(1+\beta)}\left(\frac{\beta \alpha}{1+\beta}  + (1 + \alpha + r) \ln \frac{(1+\beta)(1 + r)}{1 + \alpha + r} - \Big(1 + r - \frac{1 + \alpha + r}{1+\beta}\Big)\right) &&\\
	&& &=\frac{\eta}{\ln(1+\beta)}\left(\alpha  - \frac{\beta(1 + r)}{1 + \beta} + (1 + \alpha + r)\ln\frac{(1+\beta)(1 + r)}{1 + \alpha + r}\right). &&
\end{flalign*}\vspace*{-30pt}

\hfill
\end{proof}
			
		\Qrratio*
		\begin{proof}
			First consider the function $Q^\circ$ over the domain $r \in [\beta - \alpha, \beta]$.  We have
		\begin{align*}
			\frac{\sfd Q^\circ}{\sfd r} &= \frac{\eta}{\ln(1+\beta)} \left(\ln \frac{(1+\beta)(1+r)}{1 + \alpha + r} + (1+\alpha + r)\left(\frac1{1+r} - \frac{1}{1+\alpha + r}\right) \right)\\
			&=\frac{\eta}{\ln(1+\beta)} \left(\ln \frac{(1+\beta)(1+r)}{1 + \alpha + r} + \frac{\alpha}{1+r} \right).
		\end{align*}
		We then compare $\frac{\sfd Q^\circ}{\sfd r}$ with $\frac {Q^\circ}r$. 
		\begin{align*}
			\frac{\sfd Q^\circ}{\sfd r} - \frac {Q^\circ}r &= \frac{\eta}{\ln(1+\beta)}\left(\big(1 - \frac{1 + \alpha + r}{r}\big)\cdot \ln \frac{(1+\beta)(1+r)}{1 + \alpha + r} + \frac{\alpha}{1 + r} + \frac{\beta - \alpha}{r}\right)\\
			&= \frac{\eta}{\ln(1+\beta)}\left(- \frac{1 + \alpha}{r}\cdot \ln \frac{(1+\beta)(1+r)}{1 + \alpha + r} + \frac{\alpha}{1 + r} + \frac{\beta - \alpha}{r}\right)\\
			&\geq \frac{\eta}{\ln(1+\beta)} \left(-\frac{1+\alpha}{\beta - \alpha}\ln\frac{(1+\beta)^2}{1+\beta - \alpha} + \frac{\alpha}{1+\beta} + \frac{\beta - \alpha}{\beta}\right) > \frac{0.7\eta}{\ln(1+\beta)} > 0.
		\end{align*}
		
		Therefore, $\frac{Q^\circ(r)}{r}$ attains its minimum when $r = \beta - \alpha$ in the domain $[\beta - \alpha, \beta]$. In this domain, we have
		\begin{align*}
			\frac{Q^\circ(r)}{r} &\geq \frac{\eta}{\ln(1+\beta)(\beta - \alpha)} \left[(1 + \beta) \ln(1 + \beta - \alpha) - (\beta - \alpha)\right]\\
			&=\frac{\eta}{\ln(1+\beta)} \left[\frac{1 + \beta}{\beta - \alpha}\cdot\ln(1 + \beta - \alpha) - 1\right] \geq 0.1.
		\end{align*}

		Now we consider the domain $r \in [\beta, \infty)$.  As before, we compute $\frac{\sfd Q^\circ}{\sfd r}$:
		\begin{align*}
			\frac{\sfd Q^\circ}{\sfd r} &= \frac{\eta}{\ln(1+\beta)}\left( - \frac{\beta}{1+\beta} + \ln\frac{(1+\beta)(1+r)}{1+\alpha + r} + (1+\alpha+r)\left(\frac1{1+r} - \frac{1}{1+\alpha+r}\right)\right)\\
			&=\frac{\eta}{\ln(1+\beta)}\left( - \frac{\beta}{1+\beta} + \ln\frac{(1+\beta)(1+r)}{1+\alpha + r} + \frac{\alpha}{1 + r}\right).
		\end{align*}
		Again, we compare $\frac{\sfd Q^\circ}{\sfd r}$ with $\frac {Q^\circ}r$: 
		\begin{align*}
			\frac{\sfd Q^\circ}{\sfd r} - \frac {Q^\circ}r &= \frac{\eta}{\ln(1+\beta)}\left( - \frac{\beta}{1+\beta} + \frac{\alpha}{1+r} - \frac{\alpha}{r} + \frac{\beta(1+r)}{(1+\beta)r} + \Big(1 - \frac{1+\alpha + r}{r}\Big) \ln\frac{(1+\beta)(1+r)}{1+\alpha + r}\right)\\
			&= \frac{\eta}{\ln(1+\beta)}\left(\frac{\alpha}{1+r} - \frac{\alpha}{r} + \frac{\beta}{(1+\beta)r} - \frac{1+\alpha}{r} \ln\frac{(1+\beta)(1+r)}{1+\alpha + r}\right)\\
			&\leq \frac{\eta}{r\ln(1+\beta)}\left(\frac{\beta}{1+\beta} - (1+\alpha)\ln\frac{(1+\beta)^2}{1+\alpha + \beta}\right) \leq 0.
		\end{align*}
		So, $\frac{Q^\circ}{r}$ is minimized when $r$ tends to $\infty$. So, when $r \geq \beta$, we have 
		\begin{flalign*}
			&&\frac{Q^\circ(r)}{r} &\geq  \frac{\eta}{\ln(1+\beta)} \left(- \frac{\beta}{1+\beta} + \ln(1+\beta)\right) = \eta - \frac{\eta\beta}{(1+\beta)\ln(1+\beta)} \geq 0.1. && 
		\end{flalign*}	\vspace*{-25pt}
		
		\hfill
		\end{proof}

		\boundcstar*
		\begin{proof}  
			If $o \leq \beta - \alpha$, we have
			\begin{align*}
				\frac1o\int_0^o Q'(r)\sfd r &= \frac1o \int_0^o\frac{\eta}{\ln(1+\beta)}\left[(1+\alpha + r)\cdot \ln (1+r) - r\right] \sfd r\\
				&=\frac{\eta }{o\ln(1+\beta)} \cdot \frac{2\left(r+1\right)\left(r+2\alpha+1\right)\ln\left(r+1\right)-r\cdot\left(3r+4\alpha+2\right)}{4} \Big|_0^o\\
				&= \frac{\eta }{o\ln(1+\beta)} \cdot \frac{2\left(o+1\right)\left(o+2\alpha+1\right)\ln\left(o+1\right)-o\cdot\left(3o+4\alpha+2\right)}{4}\\
				&= \frac{\eta }{\ln(1+\beta)} \cdot \left(\frac{2\left(o+1\right)\left(o+2\alpha+1\right)\ln\left(o+1\right)}{4o}-\frac{3o+4\alpha+2}{4}\right).
			\end{align*}
	
			So, we need to compute the maximum of 
			\begin{align*}
				\frac{1}{1+o}\cdot \left[1+\alpha + \frac{3o}{2} - \frac{\eta}{\ln(1+\beta)} \times \left(\frac{2\left(o+1\right)\left(o+2\alpha+1\right)\ln\left(o+1\right)}{4o}-\frac{3o+4\alpha+2}{4}\right)\right].
			\end{align*}
	
	    The maximum of the function over $o \in [0, \beta - \alpha]$ is at most $1.445 < 1.45$, achieved when $o \approx 8.162$ (See Figure~\ref{fig:o-graph} in Appendix~\ref{appendix:wC-plotting}). \medskip
	    
		Then we consider the case $o \geq \beta - \alpha$. We define $\lambda$ such that 
			$\displaystyle \int_0^{\beta - \alpha}Q'(r)\sfd r = \frac{\lambda(\beta - \alpha)^2}{2}$.
		Indeed one can compute the value of $\lambda$ and see that $\lambda < 0.1$. To avoid the computation, we see that if $\lambda \geq 0.1$, then we have $\frac1o \int_0^o Q'(r) \sfd r \geq \frac1o\left(\frac{0.1(\beta - \alpha)^2}{2} + \int_{\beta - \alpha}^o 0.1 r \sfd r\right) \geq \frac{0.1o}{2} = 0.05o$. The quantity we need to bound is at most $\frac{1}{1+o}\cdot (1+\alpha+\frac{3o}{2} - 0.05 o ) \leq \max\{1 + \alpha, 1.45\} = 1.45$. 
		
		So we assume $\lambda < 0.1$. We have
		\begin{align*}
			&\quad \frac{1}{1+o}\left(1 + \alpha + \frac{3o}2 - \frac1o \int_0^o Q'(r)\sfd r\right) \leq \frac{1}{1+o}\left(1 + \alpha + \frac{3o}2 - \frac1o \left(\frac{\lambda(\beta - \alpha)^2}{2} + \int_{\beta - \alpha}^o 0.1 r\sfd r\right)\right)\\
			&= \frac1{1+o}\left(1 + \alpha + \frac{3 o}{2} - \frac{1}{o}\left(\frac{\lambda (\beta-\alpha)^2}{2} + \frac{0.1(o^2 - (\beta-\alpha)^2)}{2}\right) \right) \\
			&= \frac1{1+o}\left(1 + \alpha + 1.45o  + \frac{(0.1 - \lambda)(\beta-\alpha)^2}{2o} \right)\\
			&= 1.45 + \frac{1}{o}\cdot \frac{(0.1-\lambda)(\beta - \alpha)^2}{2} + \frac{1}{1 + o}\left(1 + \alpha - 1.45 - \frac{(0.1-\lambda)(\beta - \alpha)^2}{2}\right)\\
			&=  1.45 + \frac{1}{o}\cdot \frac{(0.1-\lambda)(\beta - \alpha)^2}{2} - \frac{1}{1 + o}\left(0.45 - \alpha + \frac{(0.1-\lambda)(\beta - \alpha)^2}{2}\right).
		\end{align*}
		
		To obtain the second-to-last equality, we replace $\frac{o}{1+o}$ with $1 - \frac{1}{1 + o}$ and $\frac{1}{o(o+1)}$ with $\frac{1}{o} - \frac{1}{1+o}$.
		Let $A = \frac{(0.1-\lambda)(\beta - \alpha)^2}{2}$ and $B = 0.45 - \alpha + \frac{(0.1-\lambda)(\beta - \alpha)^2}{2}$ so that the above quantity is $1.45 + \frac{A}{o} - \frac B{1+o}$. Notice that $0 < A < B$. The derivative of the function w.r.t $o$ is $-\frac{A}{o^2} + \frac{B}{(1+o)^2}$. The derivative is $0$ if $o = 1/(\sqrt{B/A}-1)$, positive if $o > 1/(\sqrt{B/A}-1)$, and negative if $0 < o < 1/(\sqrt{B/A}-1)$. So $o = 1/(\sqrt{B/A}-1)$ is a local minimum.  Therefore, the maximum of the function occurs when $o = \beta - \alpha$ or $o = \infty$.  The former case is already covered in the case $o \in [0, \beta - \alpha]$. In the later case, the function value is $1.45$.  \hfill
	\end{proof}

 \newpage
	\section{Plotting of Functions Used in the Weighted Completion Time Result} \label{appendix:wC-plotting}
	    In this section, we plot the functions used in the analysis of the algorithm for weighted completion time. See Figures~\ref{fig:rho-graph} to \ref{fig:o-graph}. 
        \begin{figure}[h!]
		    \begin{minipage}[b]{0.45\textwidth}
		        \includegraphics[width=0.95\textwidth]{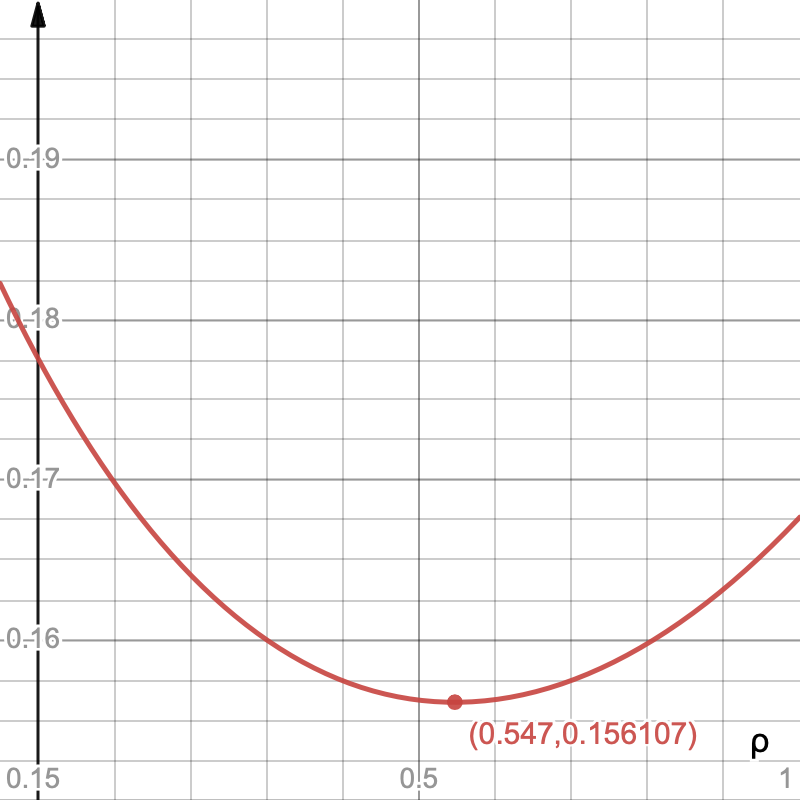}
		        \caption{The function $\frac49\cdot \left(\frac{1-e^{\rho - 1}}{1-\rho}\right)^2 \cdot \left(\frac{1}{{2\rho}} - \frac{1-e^{-{4\rho}}}{8\rho^2}\right)$ over $\rho \in [0, 1]$.}
		        \label{fig:rho-graph}
		    \end{minipage}\hfill    
    	\begin{minipage}[b]{0.45\textwidth}
				\centering
				\includegraphics[width=0.95\textwidth]{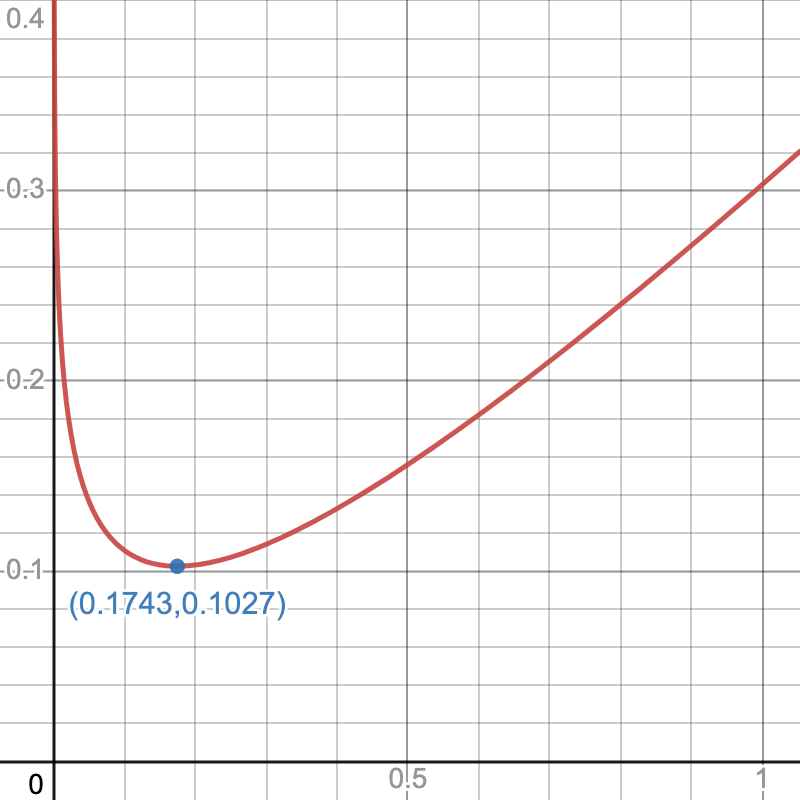}
				\caption{The function $\alpha x + \frac{\eta}{\ln(1+\beta)}\Big((1 - \alpha x)\ln\frac{1 - \alpha x}{x} - 1 + (1+\alpha)x\Big)$ over $x \in [0, 1]$.} \label{fig:plotting-3}
		\end{minipage}\medskip
	
	\begin{minipage}[b]{0.45\textwidth}
		\centering
		\includegraphics[width=0.95\textwidth]{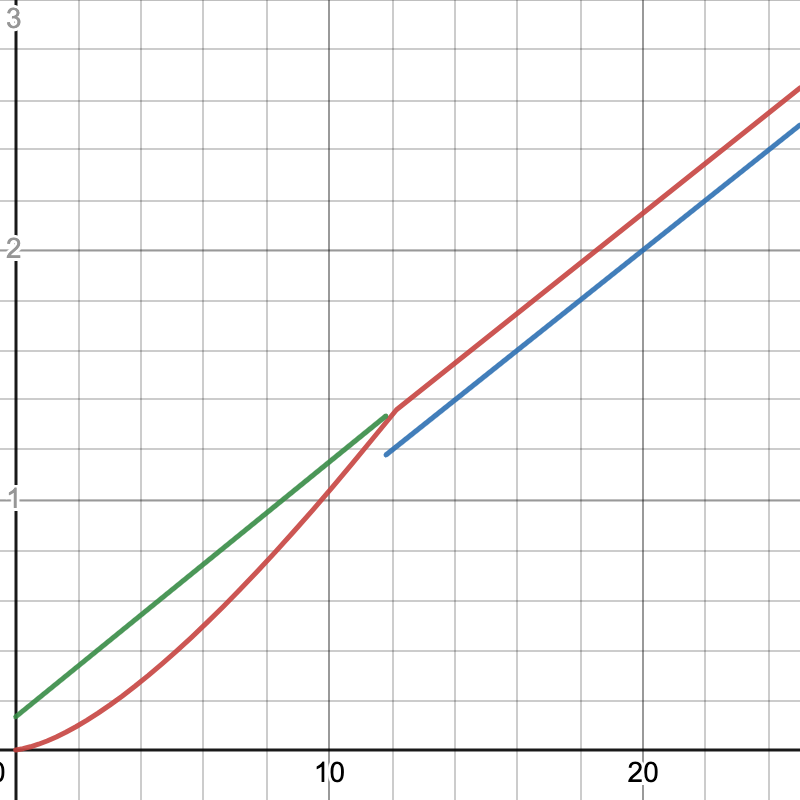}
		\caption{The function $Q^\circ(r)$ over $r \geq 0$ (the red line), the function $0.102(1+\alpha + r)$ over $r \in [0, \beta - \alpha)$ (the green line) and the function $0.1 r$ over $r \in [\beta - \alpha, \infty)$ (the blue line).}
		\label{fig:Qr}
	\end{minipage}\hfill
		    \begin{minipage}[b]{0.45\textwidth}
	\centering
	\includegraphics[width=0.95\textwidth]{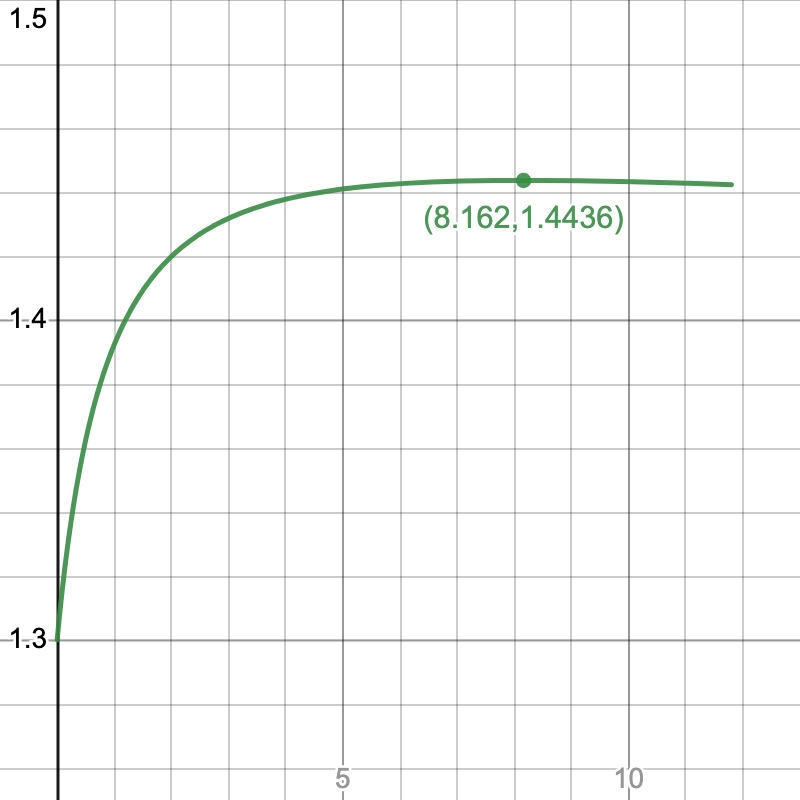}
	\caption{The function $\frac{1}{1+o}\cdot \Big[1+\alpha + \frac{3o}{2} - \frac{\eta}{\ln(1+\beta)} \times \left(\frac{2\left(o+1\right)\left(o+2\alpha+1\right)\ln\left(o+1\right)}{4o}-\frac{3o+4\alpha+2}{4}\right)\Big]$ over $o \in [0, \beta - \alpha] = [0, 11.8]$.}
	\label{fig:o-graph}
\end{minipage}\bigskip
        \end{figure}

\section{Approximation Ratio for Minimizing $L_2$-norm of Machine Loads}
\label{sec:k-2}

In this section we analytically show that the approximation ratio of our randomized algorithm for minimizing the $L_2$-norm of machine loads is $\sqrt {4/3}$. To show this, it suffices to show that $\sup_{m \in (0, 1], y \in [0, 1]} g(m, y)  = 4/3$ when $k = 2$. For convenience we reproduce $g(m, y)$ here when $k = 2$. 

$$
g(m, y) := \frac{ m \cdot (1+ (1 - m)y)^2 + (1- m) \cdot ((1- m)y)^2}{m  + (1 - m) \cdot y^2}
$$

We first check the boundary cases. For any fixed $y$, $\lim_{m \rightarrow 0} g(m, y) = 1$. Further, $g(1, y) = 1$ for any $y \in [0, 1]$ and $g(m, 0) = 1$ for any $m \in (0, 1]$. Finally, $g(m, 1) = m( 2 - m)^2 + (1 - m)^3$ and we can show that it is maximized when $m = 1/2$ and $g(1/2, 1) = 9/8$.

To consider the local optima when $m \in (0, 1), y \in (0, 1)$, we compute the partial derivative. For convenience, consider 

$$
- g(m, y) + 1 = \frac{m (1 - m) y ( y - 2)}{(1 - m) y^2 + m}
$$

Let $p$ and $q$ be the numerator and denominator respectively. Consider the partial derivative w.r.t. $m$. By simplifying 
\begin{align*}
    0 = \frac{\partial p}{\partial m} \cdot q - \frac{\partial q}{\partial m} \cdot p
    &=  ( 1- 2m) y ( y -2) ( (1 - m)y^2 + m) - ( -y^2 + 1)  m ( 1- m) y ( y - 2),
\end{align*}
we obtain 
\begin{equation}
    \label{eqn:partial-1}
    y = \frac{ m}{ 1 - m}
\end{equation}

Consider the partial derivative w.r.t. $y$. By simplifying 
\begin{align*}
    0 = \frac{\partial p}{\partial y} \cdot q - \frac{\partial q}{\partial y} \cdot p
    &=  m ( 1- m) (2y - 2) ( (1 - m)y^2 + m) - m ( 1- m) y(y - 2) 2y ( 1- m),
\end{align*}
we obtain, 
\begin{equation}
    \label{eqn:partial-2}
    \frac{1 - m}{m} = \frac{1- y}{y^2} 
\end{equation}

From Eqn~(\ref{eqn:partial-1}) and (\ref{eqn:partial-2}), we have 
$m = 1/3$ and  $y = 1/2$. This is the unique local optimum when  $m, y \in (0, 1)$ and $g( 1/3, 1/2) = 4/3$. Thus, we conclude 
$\sup_{m \in (0, 1], y \in [0, 1]} g(m, y)  = 4/3$.

\end{document}